\documentclass[lettersize,journal]{IEEEtran}
\usepackage{amsmath,amsfonts}
\usepackage{algorithmic}
\usepackage{algorithm}
\usepackage{array}
\usepackage[caption=false,font=normalsize,labelfont=sf,textfont=sf]{subfig}
\usepackage{textcomp}
\usepackage{stfloats}
\usepackage{url}
\usepackage{verbatim}
\usepackage[numbers,sort&compress]{natbib}
\usepackage{bm}
\usepackage{amssymb} 
\usepackage{multirow} 
\usepackage{booktabs}
\usepackage{amsmath,amsthm}
\usepackage{graphicx}  
\usepackage{xcolor}
\newtheorem{thm}{Theorem}
\newtheorem{definition}{Definition}

\newtheorem{assumption}{Assumption}

\newtheorem{problem}{Problem}
\newtheorem{remark}{Remark}

\begin{document}

\title{Synthesizing Safety in Infinite-Horizon Optimal Control for Disturbed High-Relative-Degree Systems via Barrier-Regulating Auxiliary Variables}

\author{Zhanglin Shangguan, Wei Xiao, Qi Li, Bo Yang, and Xinping Guan, \IEEEmembership{Fellow IEEE}
\thanks{This work was supported by the National Natural Science Foundation of China under Grant 62325306 and Grant 62273237. (\textit{Corresponding author: Bo Yang})}
\thanks{Z. Shangguan, Q. Li, B. Yang, and X. Guan are with the School of Automation and Intelligent Sensing, Shanghai Jiao Tong University, Shanghai, China, and also with the Key Laboratory of System Control and Information Processing, Ministry of Education of China, and the Shanghai Engineering Research Center of Intelligent Control and Management, Shanghai, China (e-mail: ditto331@sjtu.edu.cn; lqecho@sjtu.edu.cn; bo.yang@sjtu.edu.cn; xpguan@sjtu.edu.cn).}
\thanks{Wei Xiao is with the Robotics Engineering Department, Worcester Polytechnic Institute, Worcester, MA, USA, and also with the Computer Science and Artificial Intelligence Laboratory, Massachusetts Institute of Technology, Cambridge, MA, USA (e-mail: weixy@mit.edu).}}

\maketitle

\begin{abstract}
Optimal stabilization of safety-critical nonlinear systems requires balancing long-term performance and strict safety constraints. Existing quadratic-programming-based control barrier function (CBF) safety filters are point-wise and may exhibit myopic behavior and local trapping when the safeguarding action conflicts with the nominal optimal control. This paper develops a safety-aware infinite-horizon optimal control framework by embedding a barrier-Lyapunov function (BLF)-based safeguarding action into the system dynamics and introducing a barrier-regulating auxiliary variable, thereby reformulating the original constrained problem as an unconstrained one on an extended state space. To mitigate local trapping, we introduce an adaptive alignment-conditioned tangential excitation orthogonal to the safety direction, with activation adaptively modulated by the degree of directional alignment between the nominal and safeguarding controllers, and incorporate it as an admissible $\mathcal{L}2$ disturbance in an $H\infty$ formulation. For high-relative-degree systems under disturbances, we further augment the recursive high-order safe-set construction with barrier compensation terms to obtain a high-order BLF and formulate an adversarial disturbance attenuation problem, which is approximately solved via safe-exploration-enhanced online critic learning. Simulations demonstrate reduced local trapping, improved safety--performance trade-offs, and safe operation under disturbances.
\end{abstract}


\section{Introduction}
\IEEEPARstart{M}{odern} control systems operate in uncertain environments where safety, control performance, and disturbance rejection must be addressed in a unified manner, particularly in safety-critical applications such as aerospace and autonomous systems. To synthesize safety and stabilization objectives, control Lyapunov function–control barrier function (CLF–CBF) frameworks have been widely adopted to enforce safety constraints through point-wise trade-offs \cite{ames2016control}. However, due to their inherently myopic optimization nature, such methods may suffer from local trapping and fail to capture long-term performance considerations \cite{reis2020control,tan2024undesired}. These limitations motivate an infinite-horizon perspective, where long-term behavior is explicitly optimized via Hamilton–Jacobi–Bellman (HJB) formulations. Reinforcement learning (RL), particularly Lyapunov-based actor–critic architectures rooted in adaptive control theory, has been extensively studied as a data-driven approach for approximating HJB solutions with stability guarantees \cite{vamvoudakis2010online,bhasin2013novel,kamalapurkar2016efficient}. Nevertheless, existing RL-based optimal control methods primarily emphasize stability and performance, and generally lack mechanisms to enforce state safety constraints throughout the learning process. As a result, integrating rigorous safety guarantees into learning-based infinite-horizon optimal control remains an open problem.

Aligned with this perspective, this work addresses three fundamental challenges. 
First, incorporating safety constraints into infinite-horizon optimal control introduces a nontrivial trade-off between long-term performance optimization and safety enforcement. 
Second, barrier-based safeguarding mechanisms may induce locally trapped behaviors, where the system stagnates in certain regions of the state space despite remaining safe, thereby preventing convergence to the desired objective. 
Third, extending safety guarantees to disturbed nonlinear systems with high relative degree requires robustness-enhanced high-order safeguarding controllers that ensure disturbance rejection and closed-loop stability. 
Addressing these challenges calls for safe RL architectures that explicitly account for system structure, exploration behavior, and external disturbances while preserving both performance optimality and provable safety.

\textbf{Related works.}
Existing safe reinforcement learning approaches for continuous-time control follow multiple design paradigms depending on how safety constraints are enforced. 
One representative line of work incorporates safety directly into the optimal control formulation by augmenting the value function with barrier terms. 
In particular, reciprocal barrier functions (RBFs) have been embedded into the value function to transform state-constrained optimal control problems into unconstrained ones, enabling the application of standard RL solution methods \cite{yang2020safe,cohen2020approximate,wang2023safe}. 
However, as discussed in \cite{mahmud2021safety,cohen2023safe}, the inherent coupling between safety guarantees and performance learning in this approach introduces critical vulnerabilities. 
Specifically, the RBF-enhanced value function may lose continuous differentiability under certain system dynamics and safety constraints, potentially leading to safety violations during deployment.

To strictly enforce safety constraints, control barrier function (CBF)-based safeguarding controllers have been widely used to filter RL-based control policies \cite{cheng2019end,cheng2023safe}. 
Nevertheless, using CBFs as a filtering mechanism results in point-wise optimal control and therefore suffers from intrinsic myopic behavior \cite{krstic2023inverse}. 
To mitigate this issue, \cite{cohen2023safe} introduces a Lyapunov-like CBF (LCBF) by shifting a reciprocal CBF (RCBF) to the equilibrium point, enabling safety enforcement to align with stabilization objectives. 
Owing to the unboundedness of RCBFs near the safety boundary, the resulting safeguarding controller provides tunable and minimally invasive safety intervention. 
Building on this idea, \cite{bandyopadhyay2023lagrangian} generalizes LCBFs into barrier-Lyapunov functions (BLFs) and embeds the BLF-induced safety constraint into the value function through a Lagrangian formulation. 
In this framework, a state-dependent Lagrange multiplier is learned online using an actor–critic–Lagrangian architecture, enabling adaptive regulation of the safety–performance trade-off. 
The work in \cite{wang2025learning} further extends LCBF-based safeguarding controllers to higher-relative-degree systems by introducing a gradient similarity condition to construct adaptive dynamics for the trade-off parameter. 
However, the resulting safeguarding mechanism becomes ineffective in the presence of system uncertainties or external disturbances. 
In contrast, although \cite{bandyopadhyay2023lagrangian} is restricted to relative-degree-one systems, its formulation naturally accommodates robustness.

Despite these developments, an additional challenge arises when barrier-based safeguarding interacts with optimal stabilization objectives. 
Specifically, the closed-loop system may become locally trapped in certain regions of the state space while remaining safe, preventing convergence toward the desired objective. 
Recent studies have attempted to address related issues from different perspectives. 
The work in \cite{wang2025learning2} incorporates CBF-based safeguarding into an infinite-horizon HJB performance objective to alleviate the myopic nature of point-wise safety filtering. 
However, whether this formulation can avoid locally trapped behaviors depends on whether the value function successfully learns safety information embedded in the optimization objective. 
Alternatively, \cite{gonccalves2024control} introduces a circulation field orthogonal to the CBF gradient to drive the system away from such stagnation regions. 
While effective in certain settings, the proposed circulation mechanism does not provide explicit conditions ensuring compatibility with nonlinear system dynamics, and may potentially compromise closed-loop stability. 

\textbf{Contributions.} This paper develops a safety-aware infinite-horizon optimal control framework for nonlinear systems, together with an adaptive angle-dependent tangential excitation mechanism and a robust high-relative-degree extension. The main contributions are summarized as follows:
\begin{enumerate}
\item \textbf{Safety-aware infinite-horizon reformulation with barrier regulation.}
We embed the BLF-based safeguarding action into the system dynamics and introduce a barrier-regulating auxiliary variable, which transforms the safety-constrained infinite-horizon optimal control problem into an unconstrained problem on an extended state space. This construction enables adaptive adjustment of the safeguarding authority while preserving the infinite-horizon optimal control structure.
\item \textbf{Adaptive alignment-conditioned tangential excitation for mitigating local trapping.}
We design a tangential excitation orthogonal to the safety direction, whose activation level is adaptively modulated by the directional alignment between the nominal optimal controller and the safeguarding controller. The resulting tangential motion does not directly affect the barrier decrease condition, helps the system escape local traps in nonconvex safe sets, and can be incorporated as an admissible $\mathcal{L}2$ disturbance in the $\mathcal{H}\infty$ formulation.
\item \textbf{Robust extension to disturbed high-relative-degree systems.}
We augment the recursive high-order safe-set construction with barrier compensation terms to handle disturbances, yielding a high-order BLF for robust safety enforcement. The problem is recast as a zero-sum differential game and approximately solved via safe-exploration-enhanced online critic learning.
\end{enumerate}

\textbf{Notation.} A continuous function $\alpha:\mathbb{R}_{\ge 0}\mapsto\mathbb{R}_{\ge 0}$ is called a class $\mathcal{K}_{\infty}$ function if it is strictly increasing and $\alpha(0)=0,\ \lim_{r\mapsto\infty}\alpha(r)=\infty$. The operator $\Vert \cdot \Vert$ denotes the Euclidean norm; given the compact set $\mathcal{X}\subset\mathbb{R}^n$ and a continuous mapping $(\cdot):\mathcal{X}\mapsto \mathbb{R}^N$, define $\Vert (\cdot) \Vert_{\infty}\triangleq\sup_{x\in\mathcal{X}} \Vert(\cdot)\Vert$; $\lambda_{\max}(M)$, $\lambda_{\min}(M)$ return the maximum and minimum eigenvalues of a matrix $M$, respectively; $\mathbb{I}_n \in \mathbb{R}^{n \times n}$ represents an $n \times n$ identity matrix; define $\mathcal{B}_r(x_0)\triangleq \{x\in\mathbb{R}^n | \Vert x-x_0 \Vert\leq r \}$ as a closed ball centered at $x_0\in\mathbb{R}^n$ with radius $r>0$; the Lie derivative of a differentiable function $V:\mathbb{R}^n\mapsto\mathbb{R}$ along a vector field $\zeta:\mathbb{R}^n\mapsto\mathbb{R}^n$ is defined as $\mathcal{L}_{\zeta}V(\cdot)\triangleq \nabla V(\cdot)^\top\zeta(\cdot)$.

\section{Preliminaries and Problem Formulations}
We consider a control-affine nonlinear system with parametric uncertainty given by
\begin{flalign}\label{system}
& \dot{x} = f(x) + g(x)u + \omega(x)d(x), 
\end{flalign}
where $x\in\mathcal{X}\subseteq\mathbb{R}^n$ and $u\in\mathcal{U}\subseteq\mathbb{R}^m$ denote the system state and control input, respectively, $\omega(x)\in\mathbb{R}^{n\times k}$, and $d(x)\in\mathbb{R}^k$ denotes the unmatched external disturbance. The nonlinear functions $f:\mathcal{X}\to\mathbb{R}^n$ and $g:\mathcal{X}\to\mathbb{R}^{n\times m}$ are assumed to be Lipschitz continuous on $\mathcal{X}$, such that system \eqref{system} admits a unique solution under any continuous control input $u:\mathbb{R}_{\geq 0}\to\mathcal{U}$.

\textbf{Safety.} A system is considered safe if its states remain within a user-defined safe set $\mathcal{S}\subseteq \mathbb{R}^n$. This set is characterized by a continuously differentiable function $h:\mathbb{R}^n\mapsto\mathbb{R}$, such that:
\begin{flalign}
\mathcal{S} &= \{x\in\mathbb{R}^n : h(x)\ge 0 \}, \\
\partial\mathcal{S} &= \{x\in\mathbb{R}^n : h(x)=0 \}, \\
{\rm Int}(\mathcal{S}) &= \{ x\in\mathbb{R}^n : h(x)>0 \}, 
\end{flalign}
where $\partial\mathcal{S}$ denotes the boundary of $\mathcal{S}$, and ${\rm Int}(\mathcal{S})$ represents its interior. Moreover, outside the safe set, $h(x)<0$.

To ensure that a system remains within $\mathcal{S}$ for all time (i.e., forward invariance), Control Barrier Functions (CBFs) provide a systematic control design framework. 
\begin{definition}
\cite{ames2016control} A differentiable function $h:\mathcal{X}\to\mathbb{R}$ is called a CBF for the system \eqref{system} over $\mathcal{S}$ if, for a \( \mathcal{K}_{\infty} \) function $\alpha$ and all $x\in\mathcal{X}$, the following condition holds: 
\begin{flalign}
& \sup_{u\in\mathcal{U}}\Big\{\mathcal{L}_{f}h(x) + \mathcal{L}_{g}h(x)u \Big\}\ge -\alpha(h(x)). 
\end{flalign}
\end{definition}
Many practical constraints exhibit
\emph{high relative degree} \(r\ge 2\); that is, the control input \(u\) does
not appear in the time–derivatives of the safety function \(h(x)\) up to order
\(r-1\). To handle this case, we recursively construct auxiliary functions and
the associated high-order safe set \cite{xiao2021high}.

\begin{definition}\label{HOCBF_def}
\cite{xiao2021high} For a continuously differentiable function \(h:\mathbb{R}^n\to\mathbb{R}\) with relative degree
\(r\ge 2\), define the sequence
\begin{subequations}
\begin{flalign}
\psi_0(x) &= h(x),\\
\psi_k(x) &= \dot{\psi}_{k-1}(x)
           + \alpha_k\!\left(\psi_{k-1}\right),\; k=1,\dots,r-1, 
\end{flalign}
\end{subequations}
where each \(\alpha_k\) is a class \(\mathcal{K}_\infty\) function.
The \emph{high-order safe set} is then
\begin{flalign}
& \mathcal{D}_r
:=\bigl\{x\in\mathcal{X}:\,
\psi_k(x)\ge 0,\ \forall k=0,\dots,r-1
\bigr\}. 
\end{flalign}
Since \(\psi_{r-1}\) has relative degree one with respect to the input, we can directly apply the standard (relative-degree-one) CBF condition to it:
\begin{flalign}
& \mathcal{L}_f \psi_{r-1}(x)+\mathcal{L}_g\psi_{r-1}(x)u
+ \alpha_r\big(\psi_{r-1}(x)\big) \ge 0, 
\label{HOCBF}
\end{flalign}
for some class \(\mathcal{K}_\infty\) function \(\alpha_r\).
Any function \(h(x)\) for which there exist functions
\(\alpha_1,\dots,\alpha_r\) such that the inequality
\eqref{HOCBF} holds for all \(x\in\mathcal{X}\)
is called a \emph{high-order control barrier function} (HOCBF).
\end{definition}

Condition~\eqref{HOCBF} is affine in the control input $u$, as the control appears only through the term $\mathcal{L}_g\mathcal{L}_f^{\rho-1} h(x)\,u$. This affine structure enables the direct synthesis of a point-wise optimal safety filter via quadratic programming (QP). Specifically, given a nominal control input $u_{\mathrm{no}}(x)$ designed for performance, the safety-critical control input can be obtained by solving, at each state $x$, the following QP:
\begin{flalign}
&u_{\rm pw}(x) = \arg\min_{u\in\mathcal{U}} \|u - u_{\mathrm{no}}(x)\|^2 \notag\\
&\mathrm{s.t.}\quad 
 \mathcal{L}_f \psi_{r-1}(x)+\mathcal{L}_g\psi_{r-1}(x)u
+ \alpha_r\big(\psi_{r-1}(x)\big) \ge 0. \notag 
\end{flalign}

\textbf{Problems:} CBFs enable point-wise optimal safety filters, typically implemented via QP, by minimally modifying nominal control inputs. While computationally efficient, such point-wise optimality is inherently myopic and does not account for infinite-horizon control performance. In contrast, infinite-horizon optimal control \cite{lewis2012optimal} captures long-term performance objectives but faces fundamental difficulties in enforcing safety, particularly for nonlinear systems with high relative degree and uncertain drifts. This raises the challenge of bridging safety and infinite-horizon control performance within a unified framework, while enabling online learning with formal safety guarantees.

\begin{figure}[t]
    \centering
    \includegraphics[width=\linewidth]{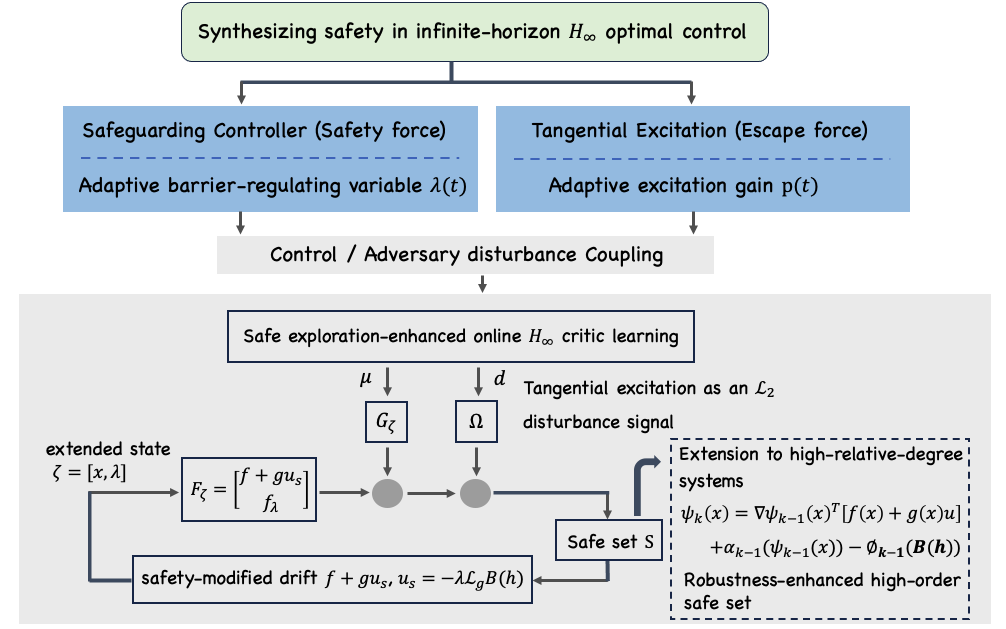}
    \caption{Barrier-regulating variable-enhanced optimal control}
    \label{structure}
\end{figure}

\section{Synthesizing Safety Filters in Infinite-Horizon Control Performance}
This section investigates the synthesis of safety filters from an infinite-horizon optimal control perspective. Starting from a basic example of an infinite-horizon optimal control problem augmented with a barrier-Lyapunov function (BLF), we elucidate the underlying principles governing the trade-off between safety enforcement and optimal stabilization. Building on this intuition, the framework is then extended to high-relative-degree systems, where safety constraints cannot be directly enforced through the control input. By drawing connections to HOCBF, we develop a systematic construction of high-order BLFs (HOBLF) that enables the realization of minimally invasive safety filters while preserving infinite-horizon control performance.

\subsection{Motivation: minimally invasive safeguarding controllers}
We consider the following state-constrained infinite-horizon optimal control problem.
\begin{problem}(Constrained Optimal Control) 
\begin{subequations}
\begin{flalign}
\min_{u\in\mathcal{U}} J(x_0) = &\int_{t_0}^{\infty} \underbrace{\left( x^\top Q x + \frac{1}{2}u^\top R u \right)}_{\ell(x,u)}\, \mathrm{d\tau}, \\
\mathrm{s.t.}\quad &\dot{x}=f(x)+g(x)u, \\
&x(t)\in\mathcal{S}, \forall t\in[t_0,\infty),  
\end{flalign}    
\end{subequations}
where $Q\succeq 0$, $R\succ 0$ and $\ell(x,u)$ is the instantaneous cost.
\end{problem}
\begin{definition}
Consider the prescribed safe set $\mathcal{S}=\{x:h(x)\ge 0\}$. Using $h(x)$, one can construct a composite function $B(h(x))$, called a barrier-Lyapunov function (BLF), such that $B(h(0))=0, B(h(x)) \ge 0, \forall x\in\mathrm{int}(\mathcal{S})$ and $\lim_{x\to\partial\mathcal{S}} B(h(x)) = \infty.$
\end{definition}
\begin{remark}
The BLF is closely related to the notion of reciprocal CBFs. In particular, both constructions enforce safety by rendering the barrier function unbounded as the system state approaches the boundary of the safe set. A key distinction is that the BLF satisfies $B(h(0))=0$, which can be interpreted as a shifted version of a reciprocal CBF that assigns zero barrier value at the equilibrium. 

By enforcing safety through the composite function $B(h(x))$, the state constraint $x(t)\in\mathcal{S}$ for all $t\ge 0$ can be equivalently encoded as the Lyapunov-like condition $\dot{B}(h(x))\le 0$. This formulation is particularly appealing, as it aligns safety enforcement with the standard Lyapunov stability paradigm, thereby providing a unified perspective in which safety and stabilization objectives are treated in a consistent manner.
\end{remark}

Let $V^*(x) = \min_{u}\int_{t}^\infty \ell(x,u) {\rm d\tau}$ be the optimal ``cost-to-go" value function. By introducing a Lagrange multiplier $\lambda\left(x(t)\right) \ge 0$ for the BLF constraint $\dot{B}\leq 0$, we can incorporate the safety into the value function as 
\begin{flalign}\label{lagrangian value function}
& V^*(x) = \min_{u}\max_{\lambda\ge 0}\int_{t}^{\infty}\Big( \ell(x,u)+\lambda(x)\frac{\rm d B}{\rm d t} \Big){\rm d}\tau. 
\end{flalign}
Based on the fundamental principle of dynamic programming \cite{kirk2004optimal}, we can obtain the stationary Hamilton-Jacobi-Bellman (HJB) equation: for any $x\in\mathcal{X}$,
\begin{flalign}\label{constrained HJB}
\min_{u}\max_{\lambda\ge 0} &\Big\{  \ell(x,u) + \nabla V^{*}(x)^\top\big[ f(x) + g(x)u \big]\notag \\
& +\lambda(x)\nabla B(x)^\top [f(x) + g(x)u ]\Big\}= 0. 
\end{flalign}
For a fixed state $x$, the terms in \eqref{constrained HJB} that depend on the control input $u$ form a quadratic function of $u$, consisting of a positive definite quadratic term and affine terms. Since the input weighting matrix satisfies $R = R^\top \succ 0$, the constrained HJB equation \eqref{constrained HJB} is strictly convex with respect to $u$, which guarantees the uniqueness of the minimizer and justifies the use of the KKT conditions \cite{almubarak2021hjb}. According to the KKT conditions, the optimal solution $\left(\lambda^*,u_{\lambda}^*\right)$ satisfies: for any $x\in\mathcal{X}$,
\begin{subequations}\label{eq:KKT_conditions}
\begin{flalign}
&\frac{\partial}{\partial u}\Big( \ell(x,u) + (\lambda^*\nabla B 
+ \nabla V^{*})^\top g u \Big)\Big|_{u = u_{\lambda}^{*}} = 0  
\label{eq:KKT_a} \\
& \lambda^{*}(x(t))\, \nabla B(x)^\top(f+g u) = 0  
\label{eq:KKT_b} \\
& \lambda^{*}(x(t)) \ge 0 
\label{eq:KKT_c}
\end{flalign}
\end{subequations}
where \eqref{eq:KKT_b} is the complementary slackness condition. Since the system is autonomous and the BLF-based safety constraint is stationary, the Lagrange multiplier $\lambda^*(x)$ is a function of the state $x$. The stationary condition \eqref{eq:KKT_a} allows us to derive the optimal control policy $u_{\lambda}^*$ as
\begin{flalign}\label{optimal safe controller}
& u_{\lambda}^*(x) = -R^{-1}\big[ \mathcal{L}_{g}V^{*}(x)+\lambda^*(x)\mathcal{L}_{g} B(x)\big]^\top. 
\end{flalign}
Substituting \eqref{optimal safe controller} into the complementary slackness condition \eqref{eq:KKT_b}, we can obtain the optimal Lagrange multiplier as
\begin{flalign}\label{optimal lamda}
& \lambda^*(x) = \max\left( \frac{-\nabla B(x)^\top(f(x)+g(x)u_{\rm o}(x))} {R_{bg}(x)}, 0\right), 
\end{flalign}
where $u_{\rm o}(x) = -R^{-1}\mathcal{L}_{g}V^*(x)^{\top}$, $R_{g}(x) = g(x)R^{-1}g(x)^\top$, and $R_{bg}(x) = \nabla B(x) R_{g}(x) \nabla B(x)^\top$.

According to \eqref{optimal safe controller}, we denote the nominal optimal controller by $u_{\rm o}(x) = -\mathcal{L}_g V^*(x)^\top$, and the safeguarding controller by $u_s(x) = - \lambda^*(x)\mathcal{L}_g B(x)^\top$. For notational simplicity, the factor $R^{-1}$ is omitted in the subsequent expressions whenever no confusion arises. The above derivation reveals that, from an infinite-horizon perspective, safety and stabilization can be systematically balanced by introducing a BLF-based Lagrange multiplier into the optimal control formulation. The optimal control input naturally decomposes into a nominal performance-oriented controller $u_{\rm o}$ and an additive safeguarding controller $u_s$, where the latter is activated only when the safety constraint becomes critical. As detailed in \cite{cohen2023safe}, the above structure suggests a key insight: exact computation of the optimal Lagrange multiplier $\lambda^*(x)$ is not essential for guaranteeing safety. Instead, it suffices to select $\lambda(x)$ as any prescribed positive constant $\lambda>0$. In this case, the safeguarding controller $u_s(x) = -\lambda\,\mathcal{L}_g B(x)^\top$ remains effective in enforcing the safety constraint, while minimally perturbing the nominal optimal controller.

To see this, consider the time derivative of the barrier-Lyapunov function $B$ under the safeguarding controller $u_s(x) = -\lambda \mathcal{L}_g B(x)^\top$, where $\lambda>0$ is a prescribed constant. The derivative of $B$ along the system trajectories satisfies
\begin{flalign}
&\dot{B}(x) 
= \nabla B(x)^\top \big( f(x) + g(x)u_s(x) \big) \notag\\
&\le \|\nabla B(x)\|\,\|f(x)\|
- \lambda\, \nabla B(x)^\top g(x) g(x)^\top \nabla B(x) \notag\\
&\le \|\nabla B(x)\|\,\|f(x)\|
- \lambda\, \underline{\sigma}\!\left(g(x)g(x)^\top\right)\|\nabla B(x)\|^2 \notag\\
&= \|\nabla B(x)\|^2
\left(
\frac{\|f(x)\|}{\|\nabla B(x)\|}
- \lambda\, \underline{\sigma}\!\left(g(x)g(x)^\top\right)
\right), 
\end{flalign}
where $\underline{\sigma}(\cdot)$ denotes the minimum eigenvalue.

By construction of the BLF, $\|\nabla B(x)\|\to\infty$ as $x\to\partial\mathcal{S}$. Consequently, for any fixed $\lambda>0$, there exists a neighborhood of the boundary $\partial\mathcal{S}$ in which $\dot{B}(x)<0$. This implies that the BLF remains bounded along system trajectories, thereby guaranteeing forward invariance of the safe set $\mathcal{S}$.

Moreover, the divergence of $\nabla B(x)$ near the constraint boundary endows the safeguarding mechanism with an inherent robustness property, as the corrective action dominates bounded model uncertainties and disturbances in the vicinity of $\partial\mathcal{S}$. As a result, safety can be enforced using an arbitrarily small but strictly positive constant $\lambda$, yielding a minimally invasive safeguarding controller that preserves the nominal infinite-horizon control performance.

\subsection{Auxiliary system for safety trade-off}
Inspired by the analysis in the previous subsection, we observe that as long as the Lagrange multiplier $\lambda(x)$ remains strictly positive for all $t>0$, the safeguarding controller term $\lambda \mathcal{L}_g B(x)$ is sufficient to guarantee the safety of the system state. Moreover, the magnitude of the Lagrange multiplier directly reflects the relative importance assigned to safety enforcement.

In practice, it is desirable to endow the Lagrange multiplier $\lambda$ with sufficient flexibility such that it can adapt according to the safety requirement while remaining consistent with an optimal performance objective. Following a similar idea to that in \cite{xiao2021adaptive}, we introduce an auxiliary dynamical system for the Lagrange multiplier in the general form
\begin{flalign}\label{auxiliary_system_general}
& \dot{\lambda} = f_{\lambda}(x,\lambda) + g_{\lambda}(x,\lambda)\, v, 
\end{flalign}
where $v \in \mathbb{R}$ is a virtual control input, and $f_{\lambda}$ and $g_{\lambda}$ characterize the intrinsic and input dynamics of the multiplier, respectively. 

Although $\lambda>0$ is sufficient from a theoretical perspective, it is beneficial in practice to introduce a strictly positive safety-related lower bound $\lambda_{\mathrm{ref}}(x)>0$. 
The function $\lambda_{\mathrm{ref}}(x)$ can be designed based on geometric margins, risk indicators, or prior engineering knowledge. 
Its purpose is to prevent $\lambda$ from becoming excessively small when the system approaches the constraint boundary, thereby enabling earlier safeguarding intervention. 
Since $\lambda$ is a design variable rather than a physical state, introducing such a reference lower bound does not alter the original safety requirement but instead provides an additional conservative margin to enhance robustness and smoothness.

Define $h_{\lambda}(x,\lambda)=\lambda-\lambda_{\mathrm{ref}}(x)$. 
Imposing the CBF condition $\dot h_\lambda+\alpha_\lambda(h_\lambda)\ge 0$ with a class-$\mathcal{K}$ function $\alpha_\lambda(\cdot)$ leads to
\begin{flalign}\label{lambda_cbf_constraint}
& f_{\lambda}(x,\lambda)
+ g_{\lambda}(x,\lambda)\, v
- \nabla_x \lambda_{\mathrm{ref}}(x)^\top
\Big( f_{x,s}(x,\lambda) + g(x)u \Big)\notag\\
&+ \alpha_{\lambda}\!\left(\lambda-\lambda_{\mathrm{ref}}(x)\right)
\ge 0, 
\end{flalign}
where $f_{x,s}(x,\lambda)=f(x)-\lambda(x) g(x)\mathcal{L}_g B(x)^\top$ denotes the safety-modified drift dynamics of the original system.

Combining the positivity constraint and \eqref{lambda_cbf_constraint}, the admissible set of the virtual input $v$ is given by
\begin{flalign}
&\mathcal{U}_{\lambda}(x,\lambda,u)
=
\Big\{
v \in \mathbb{R} \ \Big| \
 f_{\lambda}(x,\lambda)
+ g_{\lambda}(x,\lambda)\, v
- \nabla \lambda_{\mathrm{ref}}(x)^\top \notag\\
&\big( f_{x,s}(x,\lambda)+ g(x)u \big) 
+ \alpha_{\lambda}\!\left(\lambda-\lambda_{\mathrm{ref}}(x)\right)
\ge 0
\Big\}. 
\end{flalign}

Assume that the Lagrange multiplier admits a desired steady-state value $\lambda^*$. 
Let $\tilde{\lambda}=\lambda-\lambda^*$ and introduce the augmented state $\zeta = [x^\top,\tilde{\lambda}]^\top$. 
Since $\lambda^*$ is constant, we have $\dot{\tilde{\lambda}}=\dot{\lambda}$. 
By incorporating the safeguarding controller for the original state constraint $h(x)$, the safety-modified state dynamics can be written as $\dot{x}=f_{x,s}(x,\lambda)+g(x)u$, where $f_{x,s}(x,\lambda)=f(x)-\lambda(x) g(x)\mathcal{L}_g B(x)^\top$. 
Combining \eqref{auxiliary_system_general} with the safety-modified state dynamics yields the safety-modified extended system
\begin{flalign}\label{extended_system}
\dot{\zeta}
&=
\begin{bmatrix}
\dot{x} \\
\dot{\tilde{\lambda}}
\end{bmatrix}
=
\underbrace{\begin{bmatrix}
f_{x,s}(x,\lambda) \\
f_{\lambda}(x,\lambda)
\end{bmatrix}}_{F_{\zeta}(\zeta)}
+
\underbrace{\begin{bmatrix}
g(x) & 0 \\
0 & g_{\lambda}(x,\lambda)
\end{bmatrix}}_{G_{\zeta}(\zeta)}
\underbrace{\begin{bmatrix}
u \\
v
\end{bmatrix}}_{\mu}. 
\end{flalign}

We can then formulate a new infinite-horizon unconstrained optimal control problem with the above safety-modified extended system \eqref{extended_system} as follows:
\begin{problem}(Unconstrained Optimal Control) 
\begin{subequations}
\begin{flalign}\label{unconstrained_ocp}
\min_{\mu} \; J(\zeta_0) 
= &\int_{t_0}^{\infty} 
\underbrace{\Big( \mathcal{Q}_{\zeta}(\zeta)+ \mathcal{R}_{\zeta}(\mu)\Big)}_{\ell(\zeta,\mu)}
\, \mathrm{d\tau}, \\
\mathrm{s.t.}\quad 
&\dot{\zeta}=F_{\zeta}(\zeta)+G_{\zeta}(\zeta)\mu, 
\end{flalign}    
\end{subequations}
where $\mathcal{Q}_{\zeta}(\zeta) = \zeta^\top Q_{\zeta} \zeta$ with $Q_{\zeta}={\rm diag}(Q,q_{\lambda})$, and $\mathcal{R}_{\zeta}(\mu)=\frac{1}{2}\mu^\top R_{\zeta} \mu$ with $R_{\zeta} = {\rm diag}(R,r_v)$.
\end{problem}
It is worth noting that the safeguarding controller is incorporated into the drift term of the system dynamics, thereby forming a safety-modified dynamics $f_{x,s}$. This embedding allows the original constrained optimal control problem to be recast as an unconstrained one over the extended state space. Moreover, the additional state introduced in \eqref{extended_system} enables adaptive adjustment of the safeguarding controller, which provides extra flexibility in balancing safety and optimality. Consequently, the resulting problem can be solved using standard methods for infinite-horizon optimal control.

\subsection{Tangential $\mathcal{L}_2$ excitation for escaping local trapping}

Although the adaptive safeguarding gain strengthens the normal component of the control action near the safety boundary, local trapping may still arise in nonconvex safe sets. 
Such behavior typically occurs when the nominal optimal direction and the safeguarding direction become nearly aligned, resulting in a vanishing total control input even though the state is not at the desired equilibrium. Geometrically, this situation is characterized by the alignment of the control effectiveness vectors $\mathcal{L}_g V(x)$ and $\mathcal{L}_g B(x)$.

To detect this configuration, we define an alignment measure between $\mathcal{L}_g V(x)$ and $\mathcal{L}_g B(x)$ as
\begin{flalign}
\cos\theta(x)
=
\frac{\mathcal{L}_g V(x)^\top \mathcal{L}_g B(x)}
{\|\mathcal{L}_g V(x)\|\,\|\mathcal{L}_g B(x)\|}. 
\end{flalign}
When $\cos\theta(x)$ approaches $1$, the two directions become nearly collinear, indicating a risk of local trapping.
This alignment condition is used to activate a tangential excitation mechanism.

The tangential direction is constructed by projecting $\mathcal{L}_g V(x)$ onto the subspace orthogonal to $\mathcal{L}_g B(x)$. 
Let
\begin{flalign}
\Pi_\perp(x)
=
\mathbb{I}_m
-
\frac{\mathcal{L}_g B(x)\mathcal{L}_g B(x)^\top}
{\|\mathcal{L}_g B(x)\|^2}, 
\end{flalign}
then the instantaneous tangential direction is defined as
\begin{flalign}
\tilde u_t(x)
=
\Pi_\perp(x)\,\mathcal{L}_g V(x). 
\end{flalign}
By construction, $\tilde u_t(x)$ lies in the null space of $\mathcal{L}_g B(x)$ and therefore does not directly interfere with the barrier decrease condition.

To generate a finite-energy excitation, we introduce a stable first-order filter
\begin{flalign}\label{excitation_gain_p}
\dot p = -\alpha_t p + \nu(x), 
\quad
u_t
=
k_t\, p\,
\frac{\tilde u_t(x)}
{\|\tilde u_t(x)\|+\varepsilon}, 
\end{flalign}
where $\alpha_t>0$, $k_t>0$, and $\varepsilon>0$ is a small regularization constant. The activation signal $\nu(x)$ is designed based on the alignment condition between $\mathcal{L}_gV(x)$ and $\mathcal{L}_gB(x)$. 
Specifically, we define 
\begin{flalign} 
\nu(x)=\kappa\,\max\!\left\{0,\, \cos\theta_0-\cos\theta(x) \right\}, 
\end{flalign} 
where $\theta(x)\in[0,\pi]$ denotes the directional misalignment between the nominal optimal control direction and the safeguarding direction, $\theta_0\in(0,\pi)$ is a prescribed triggering threshold, and $\kappa>0$ determines the excitation intensity. Hence, the tangential excitation is activated only when $\theta(x)>\theta_0$. This corresponds to the regime in which the two directions become sufficiently conflicting, so that the nominal control and the safeguarding action tend to counteract each other and a stagnation point may arise. In practice, one may simply choose $\theta_0=\pi/2$, so that excitation is triggered whenever the directional conflict exceeds $90^\circ$. Since the filter is exponentially stable and $\nu(x)$ is bounded, the resulting tangential input satisfies $\int_0^\infty \|u_t(t)\|^2 dt < \infty$, i.e., $u_t\in \mathcal{L}_2[0,\infty)$.

The overall control law becomes $u = u_o + u_s + u_t$. 
Because $u_t$ is orthogonal to $\mathcal{L}_g B(x)$, forward invariance of the safe set is preserved provided that the safeguarding gain is sufficiently large. 
From the learning perspective, the tangential component is interpreted as a structured disturbance, and the critic update is therefore formulated under an $H_\infty$ framework. 
As the critic converges and internalizes safety information, the activation signal diminishes and the filter state decays exponentially, leading to $u_t(t)\to 0$. 
Hence, the excitation is transient and does not alter the asymptotic optimal behavior.

\section{Robustness Enhancement for Disturbed High-Relative-Degree Systems}
In this subsection, we discuss how the above analysis developed for the extended system \eqref{extended_system} can be extended to high-relative-degree systems subject to external disturbances, where the system dynamics are described by \eqref{system}. The robustness enhancement problem is investigated from two complementary perspectives. First, we focus on guaranteeing safety in the presence of disturbances, ensuring that the state remains within the admissible set despite external uncertainties. Second, we address stability preservation under disturbances, aiming to maintain desirable closed-loop performance and convergence properties of the system. 

\subsection{Robustness-enhanced high-order BLF}
We first extend the notion of BLFs to systems with high relative degree. For high-relative-degree systems, it holds that $\mathcal{L}_g B(h(x)) = 0$, which implies that the safeguarding controller constructed in the previous subsection cannot directly influence the system input and therefore becomes ineffective.

To address this limitation, we draw inspiration from the HOCBF construction introduced in Definition~\ref{HOCBF_def}. In particular, for a safety constraint $h(x)$ with relative degree $r$, the associated high-order safety set is characterized by the function $\psi_{r-1}(x)$ obtained through the recursive HOCBF construction. This formulation ensures that the control input appears explicitly in the derivative of $\psi_{r-1}(x)$, thereby enabling safety enforcement for high-relative-degree systems. A barrier-Lyapunov functional can then be constructed by applying a barrier functional to the highest-order safety function, resulting in a high-order barrier-Lyapunov function of the form $B(\psi_{r-1}(x))$. 

Due to the presence of external disturbances, the high-order safe set construction introduced in Definition~X is no longer directly applicable. If the disturbance term $d(x)$ is ignored and the safety sets are constructed solely based on the nominal dynamics $\dot{x} = f(x) + g(x)u$, potential safety violations may occur. Motivated by the inherent disturbance attenuation property of BLFs in first-order systems, we enhance the robustness of the high-order safety set construction by introducing additional barrier compensation terms at each recursive step. Specifically, the safety sets are defined as
\begin{subequations}
\begin{flalign}
&\psi_0(x) = h(x), \\
&\psi_k(x) = \nabla \psi_{k-1}(x)^\top \bigl(f(x) + g(x)u\bigr)\\
&\qquad\qquad+ \alpha_{k-1}\bigl(\psi_{k-1}(x)\bigr)
- \phi_{k-1}\bigl(B(h(x))\bigr), 
\end{flalign}    
\end{subequations}
where $\phi_{k-1}(\cdot)$ denotes an additional barrier compensation term introduced to counteract the effect of disturbances.

The intuition behind this construction is as follows. Although the disturbance term $d(x)$ is not explicitly incorporated in the recursive definition and the nominal dynamics are used to construct the high-order safety sets, the introduced barrier compensation term grows unbounded as the system state approaches the boundary of the safe set. In particular, when $x \to \partial \mathcal{S}$, the barrier-Lyapunov functional satisfies $B(h(x)) \to \infty$. Therefore, as long as the disturbance $d(x)$ is bounded, its effect can be compensated by the barrier term.

To formally justify this claim, we first consider the first-order case, from which the result can be recursively extended to the general $r$-th order construction. Assume that the disturbance satisfies $\|d(x)\| \le \|d\|_{\infty} < \infty$. Taking the time derivative of $h(x)$ along the disturbed dynamics yields
\begin{flalign}
\dot{h}(x)
&= \nabla h(x)^\top \bigl(f(x) + g(x)u + d(x)\bigr) \notag\\
&\ge \nabla h(x)^\top \bigl(f(x) + g(x)u\bigr)
- \|\nabla_x h\|\,\|d\|. 
\end{flalign}
By enforcing the nominal safety condition with barrier compensation, we further obtain
\begin{flalign}
&\dot{h}(x)
\ge -\alpha_0\bigl(h\bigr)
+ \phi_0\bigl(B(h)\bigr)
- \|\nabla h\|\,\|d\|. 
\end{flalign}
For any bounded disturbance $d(x)$, there exists a sufficiently small boundary layer of the safe set in which the barrier feedback dominates the disturbance effect. Specifically, there exists a constant $\delta>0$ such that 
$\phi_0\!\bigl(B(h(x))\bigr)-\|\nabla h(x)\|\,\|d(x)\|\ge 0$ holds for all states satisfying $0<h(x)\le\delta$. This condition implies that sufficiently close to the safety boundary, the barrier feedback remains strong enough to counteract the disturbance. As a result, the set $\mathcal{R}_{\beta}(0)$ is forward invariant, which establishes safety for the first-order case.

By recursively applying the same argument, the robustness property naturally extends to the $r$-th order safety set $\psi_r(x)$. Finally, by applying the barrier-Lyapunov functional $B(\cdot)$ to $\psi_r(x)$, a high-order BLF is obtained, which ensures robust safety for high-relative-degree systems in the presence of bounded disturbances.

\subsection{Disturbance-attenuated analysis via ISS perspective}
We further investigate the stabilization performance of the safety-modified extended system in the presence of external disturbances. 
Consider the disturbed extended dynamics given by
\begin{flalign}\label{extended_sys_disturbance}
&\dot{\zeta}
=
F_{\zeta}(\zeta)
+
G_{\zeta}(\zeta)\mu
+
\Omega(\zeta)d, 
\end{flalign}
with the extended disturbance gain $\Omega(\zeta) =
\begin{bmatrix}
\omega(x) \\
0
\end{bmatrix}$. It is noted that the disturbance affects only the original system states and does not directly act on the auxiliary safety-related state introduced by the system augmentation.

From an input-to-state stability (ISS) viewpoint, robustness against disturbances can be characterized through a dissipativity-type inequality. 
Specifically, if there exists a continuously differentiable positive definite function $V(\zeta)$ such that along the system trajectories
$\dot V(\zeta)\le -\mathcal{Q}_{\zeta}(\zeta)-\mathcal{R}_{\zeta}(\mu)+\gamma^2\|d\|^2$, where $\gamma>0$ denotes the prescribed disturbance attenuation level, 
then the closed-loop system admits an $\mathcal{L}_2$-gain bound from the disturbance input $d$ to the performance output $(\zeta,\mu)$. 
Integrating this inequality over $[t_0,\infty)$ gives
$\int_{t_0}^{\infty}(\mathcal{Q}_{\zeta}(\zeta)+\mathcal{R}_{\zeta}(\mu))\,d\tau
\le V(\zeta_0)+\gamma^2\int_{t_0}^{\infty}\|d(\tau)\|^2 d\tau$,
which shows that the disturbance energy is attenuated by the factor $\gamma$ and the system is input-to-state stable in an energy sense.

Motivated by the above ISS-based analysis, we formulate a robust optimal control problem by treating the external disturbance as an adversarial input that seeks to degrade system performance. To this end, the disturbance attenuation objective is incorporated into the cost functional by penalizing the disturbance energy.

\begin{problem}(Robust Optimal Stabilization)\label{robust pro}
\begin{subequations}
\begin{flalign}
\min_{\mu}\ \max_{d}\ 
J_{\gamma}(\zeta_0)
= &
\int_{t_0}^{\infty}
\Big(
\mathcal{Q}_{\zeta}(\zeta)
+
\mathcal{R}_{\zeta}(\mu)
-
\gamma^2 \|d\|^2
\Big)
\, d\tau, \\
\mathrm{s.t.}\quad
\dot{\zeta}
= &
F_{\zeta}(\zeta)
+
G_{\zeta}(\zeta)\mu
+
\Omega(\zeta)d. 
\end{flalign}
\end{subequations}
\end{problem}

For a given attenuation level $\gamma$, the existence of a stabilizing solution to this problem guarantees that the closed-loop extended system achieves an $L_2$-gain from the disturbance $d$ to the performance output no greater than $\gamma$. The smallest achievable value of $\gamma$ characterizes the best attainable disturbance rejection capability of the system.

Following Problem~\ref{robust pro}, define the optimal value function of the zero-sum differential game as
$V^{\ast}(\zeta)
\triangleq
\min_{\mu(\cdot)}\max_{d(\cdot)}
J_{\gamma}(\zeta,\mu(\cdot),d(\cdot))$.
Define the Hamiltonian function
\begin{flalign}
\mathcal{H}(\zeta,\mu,d,\nabla V)
\triangleq &
\mathcal{Q}_{\zeta}(\zeta)
+
\mathcal{R}_{\zeta}(\mu)
-
\gamma^2 \|d\|^2
\notag\\
&+\nabla V(\zeta)^{\top}
\big(
F_{\zeta}
+
G_{\zeta}\mu
+
\Omega d
\big). 
\label{eq:Hamiltonian_def}
\end{flalign}
Under standard regularity assumptions and the Isaacs condition, the optimal value function satisfies the stationary Hamilton--Jacobi--Isaacs (HJI) equation
$0=\min_{\mu}\max_{d}\mathcal{H}(\zeta,\mu,d,\nabla V^*(\zeta))$.
The corresponding saddle-point policies are given by
$\mu^*(\zeta)=\arg\min_{\mu}\mathcal{H}(\zeta,\mu,d,\nabla V^*(\zeta))$
and
$d^*(\zeta)=\arg\max_{d}\mathcal{H}(\zeta,\mu,d,\nabla V^*(\zeta))$.

In particular, since $\mathcal{R}_{\zeta}(\mu)=\tfrac{1}{2}\mu^{\top}R_{\zeta}\mu$ and the disturbance penalty is quadratic, the stationary conditions
$\partial \mathcal{H}/\partial \mu = 0$ and $\partial \mathcal{H}/\partial d = 0$ yield the closed-form expressions
\begin{subequations}
\begin{flalign}
\mu^{\ast}(\zeta)
&=
-
R_{\zeta}^{-1}
G_{\zeta}(\zeta)^{\top}\,
\nabla V^{\ast}(\zeta),
\label{eq:mu_star_closed_form}
\\
d^{\ast}(\zeta)
&=
\frac{1}{2\gamma^2}
\Omega(\zeta)^{\top}\,
\nabla V^{\ast}(\zeta). 
\label{eq:d_star_closed_form}
\end{flalign}
\end{subequations}

Substituting \eqref{eq:mu_star_closed_form}--\eqref{eq:d_star_closed_form} into \eqref{eq:Hamiltonian_def} leads to the equivalent nonlinear partial differential equation
\begin{flalign}
0
=&
\mathcal{Q}_{\zeta}(\zeta)
+
\nabla V^{\ast}(\zeta)^{\top}F_{\zeta}
-\nabla V^{\ast}(\zeta)^{\top}
G_{\zeta}R_{\zeta}^{-1}G_{\zeta}^{\top}
\nabla V^{\ast}(\zeta)
\notag\\ &+\frac{1}{4\gamma^{2}}
\nabla V^{\ast}(\zeta)^{\top}
\Omega(\zeta)\Omega(\zeta)^{\top}
\nabla V^{\ast}(\zeta). 
\label{eq:HJI_expanded}
\end{flalign}
Equation~\eqref{eq:HJI_expanded} reveals that the disturbance channel contributes an adversarial (maximizing) quadratic term in the gradient of $V^{\ast}$, which captures the worst-case disturbance effect and establishes the $H_\infty$-type robustness of the resulting optimal policy.

In particular, the value function $V^{\ast}(\zeta)$ is unknown \emph{a priori} and cannot be obtained in closed form for general nonlinear systems, which significantly limits the direct implementation of the optimal control law. To address this challenge, the next section develops an online approximation framework based on safe exploration--enhanced RL. By leveraging data-driven learning mechanisms together with safety-aware exploration strategies, the proposed approach seeks to approximate the optimal value function and the associated saddle-point policies, while preserving robustness and safety guarantees during the learning process

\section{Safe-Exploration-Enhanced Online Learning Implementation}
This section presents a safe exploration--enhanced online reinforcement learning framework to approximately solve the HJI equation derived in the previous section. By leveraging a barrier-augmented system formulation, the learning process is naturally constrained by a safeguarding controller, such that each instantaneous evaluation of the Hamiltonian error can be interpreted as safe exploration without violating prescribed safety specifications. The proposed approach enables online approximation of the optimal value function and the associated saddle-point policies using measured system data, while preserving robustness against disturbances throughout the learning phase. Finally, the closed-loop stability of the resulting learning-based control system is analyzed.
 
\subsection{Value function approximation}
To facilitate online implementation, the optimal value function $V^{\ast}(\zeta)$ associated with the extended system is approximated using a neural network (NN) parameterization.
By virtue of the universal approximation property of neural networks, any sufficiently smooth function defined on a compact set can be approximated to arbitrary accuracy by a linear-in-the-parameters structure with nonlinear basis functions.

Accordingly, the optimal value function $V^{\ast}(\zeta)$ is expressed as
\begin{flalign}
& V^{\ast}(\zeta)
=
W^{\ast\top}\sigma(\zeta)
+
\epsilon(\zeta), 
\end{flalign}
where $W^{\ast}\in\mathbb{R}^{J}$ denotes the ideal weight vector,
$\sigma(\zeta)\in\mathbb{R}^{J}$ is a known vector of continuously differentiable basis functions (e.g., neural network activation functions), and $\epsilon(\zeta)$ represents the bounded approximation error. Correspondingly, we can write the optimal control policy and worst-case disturbance as
\begin{subequations}
\begin{flalign}
& \mu^*(\zeta) = -R_{\zeta}^{-1}G_{\zeta}(\zeta)^\top \Big( \nabla\sigma(\zeta)^\top W^* + \nabla\epsilon(\zeta)\Big), \label{ideal_policy}\\
& d^*(\zeta) = \frac{1}{2\gamma^2} \Omega(\zeta)^\top  \Big( \nabla\sigma(\zeta)^\top W^* + \nabla\epsilon(\zeta)\Big).  \label{ideal_disturbance} 
\end{flalign}
\end{subequations}

In practice, since $W^{\ast}$ is unknown, an estimate $\hat W$ is employed to construct the approximate value function
\begin{flalign}
& \hat V(\zeta)
=
\hat W^{\top}\sigma(\zeta). 
\end{flalign}
It is assumed that the basis functions $\sigma(\zeta)$ and their gradients $\nabla\sigma(\zeta)$ are known and bounded over the domain of interest.

Based on the approximate value function $\hat V(\zeta)$, the corresponding approximate optimal control policy and worst-case disturbance are obtained by substituting $\nabla \hat V(\zeta)$ into the stationary optimality conditions derived from the HJI equation.
Specifically, the approximate control policy and the approximate worst-case disturbance are given by
\begin{subequations}
\begin{flalign}
& \hat{\mu}(\zeta)
=
-
R_{\zeta}^{-1}
G_{\zeta}(\zeta)^{\top}
\nabla \sigma(\zeta)^\top \hat W, \label{es_policy}\\
& \hat d(\zeta)
=
\frac{1}{2\gamma^{2}}
\Omega(\zeta)^{\top}
\nabla \sigma(\zeta)^\top \hat W. \label{es_disturbance}
\end{flalign}    
\end{subequations}

The above parameterizations establish a foundation for online learning of the value function and the associated saddle-point policies. In the subsequent subsection, adaptive learning laws are developed to update the weight estimate $\hat W$ based on Hamiltonian error minimization.

\subsection{Safe-exploration-enhanced adaptive critic design}
Based on the value function approximation developed in the previous subsection, this subsection presents a safe exploration--enhanced adaptive critic design for online learning of the optimal value function.

By substituting the approximate control policy $\hat{\mu}(\zeta)$ and the approximate worst-case disturbance $\hat d(\zeta)$ into the Hamiltonian function, the instantaneous Hamiltonian error is defined as
\begin{flalign}\label{be}
& \delta(\zeta,\hat{\mu},\hat d)
=
\ell(\zeta,\hat{\mu},\hat d)\notag\\
& +
\nabla \hat V(\zeta)^{\top}
\Big(
F_{\zeta}(\zeta)
+
G_{\zeta}(\zeta)\hat{\mu}(\zeta)
+
\Omega(\zeta)\hat d(\zeta)
\Big), 
\end{flalign}
where
$\ell(\zeta,\hat{\mu},\hat d)
=
\mathcal{Q}_{\zeta}(\zeta)
+
\mathcal{R}_{\zeta}(\hat{\mu})
-
\gamma^2\|\hat d\|^2$.
The Hamiltonian error quantifies the deviation from the HJI optimality condition at the current state and serves as a measurable signal for critic learning.

To enhance data richness and facilitate satisfaction of the persistent excitation (PE) condition, additional state samples are collected along exploratory trajectories.
Specifically, let $\{\zeta_{t_i}\}_{i=1}^{N_s}$ denote a set of sampled state points obtained from the system trajectory or auxiliary exploration signals within a compact domain.
For each sampled state $\zeta_{t_i}$, the corresponding Hamiltonian error is computed as
\begin{flalign}\label{sample_be}
\delta_{t_i}
=&
\ell(\zeta_{t_i},\hat{\mu}(\zeta_{t_i}),\hat d(\zeta_{t_i}))+ \nabla \hat V(\zeta_{t_i})^{\top}\cdot\notag\\
&
\Big(
F_{\zeta}(\zeta_{t_i})
+
G_{\zeta}(\zeta_{t_i})\hat{\mu}(\zeta_{t_i})
+
\Omega(\zeta_{t_i})\hat d(\zeta_{t_i})
\Big). 
\end{flalign}
These sampled Hamiltonian errors collectively provide informative regression data for updating the critic parameters, thereby improving learning efficiency and robustness.

\begin{remark}
It is worth emphasizing that the Hamiltonian error is evaluated based on the barrier-augmented extended system dynamics.
As a result, the exploration process is implicitly constrained by the safeguarding controller embedded in the system formulation, and the resulting data correspond to \emph{safe exploration} of the original system.
In contrast, computing the Hamiltonian error using the original system dynamics without safety augmentation would primarily emphasize control performance while potentially violating safety constraints.
The proposed framework allows for a flexible trade-off between safety and performance by appropriately combining Hamiltonian error evaluations obtained under the safeguarding controller and those computed from the nominal system dynamics.
\end{remark}

In line with the recursive least-squares-based method, the estimated critic weights $\hat{W}$ are updated to minimize the accumulated error $E = \int_{t_0}^t \left( \delta_t^2 + \sum_{i=1}^N\delta_{ti}^2 \right)d\tau $, as
\begin{subequations}
\begin{flalign}
&\dot{\hat{W}} = - k_{1}\Gamma\frac{\rho}{\xi^2}\delta_t -\frac{k_{2}}{N}\Gamma\sum_{i=1}^N \frac{\rho_i}{\xi_i^2}\delta_{ti}, \label{critic law}\\[3pt]
&\dot{\Gamma} = \eta\Gamma - k_{1}\Gamma\frac{\rho\rho^\top}{\xi^2}\Gamma - \frac{k_{2}}{N}\Gamma\sum_{i=1}^N \frac{\rho_i\rho_i^\top}{\xi_i^2}\Gamma, \label{Gamma law} 
\end{flalign}
\end{subequations}
where $\rho(t)=\nabla\sigma\left( F_{\zeta} + G_{\zeta}\hat{\mu} + \Omega\hat{d} \right)$ is the regressor vector, $\xi(t):=\sqrt{1+k_3 \rho^\top \rho}$ is the normalized term, $k_{1}, k_{2}>0$ and $\eta>0$ are user-defined learning rate and forgetting constant, respectively.

To guarantee the convergence of the states $x$ and estimated critic weights $\hat{W}$, the following excitation conditions, boundedness assumptions and Lipschitz continuity of the optimal controller are established.
\begin{assumption}\label{PE_condition}
\cite{kamalapurkar2016efficient} There exists a strictly positive constant $b\in\mathbb{R}_{> 0}$ such that for all $t\ge 0$, 
\begin{flalign}
& 0< b \leq \inf_{t\ge 0} \left\{\sigma_{\min}\left( \frac{k_{1}\rho\rho^\top}{\xi^2}+\frac{k_{2}}{N} \sum_{i=1}^N\frac{\rho_i\rho_i^\top}{\xi_i^2} \right)\right\}. 
\end{flalign}
\end{assumption}
\begin{remark}
This is the relaxed rank-like persistently exciting (PE) condition induced by the concurrent learning technique. It can be observed that by incorporating additional sampled state trajectories for BE extrapolation, the above rank-like PE condition can be more readily satisfied. Moreover, under this relaxed PE condition, the least-squares matrix gain $\Gamma$ is bounded as $\underline{\Gamma} \mathbb{I}_L \leq \Gamma(t) \leq \overline{\Gamma} \mathbb{I}_L$, where $\overline{\Gamma} > \underline{\Gamma}>0$.
\end{remark}

\subsection{Stability analysis}
Before we proceed to the stability analysis of the extended system \eqref{extended_sys_disturbance} under the approximate optimal control policy $\hat{\mu}$ and worst-case disturbance $\hat{d}$, we simplify the Hamiltonian error $\delta$ defined as \eqref{be}. Note that $\ell(\zeta,\mu^*,d^*) + \nabla V^*(\zeta)^\top [F_{\zeta}+G_{\zeta}\mu^*+\Omega d^*] = 0$; subtracting this equality from the Hamiltonian error \eqref{be} and following the derivations in \cite{kokolakis2022safety,kamalapurkar2016efficient} yields
\begin{flalign}\label{bee}
\delta = -\tilde{W}^\top \rho + \Delta(x), 
\end{flalign}
where the function $\Delta(x) : \mathbb{R}^n\to \mathbb{R}$ is uniformly bounded over $\mathcal{S}$ such that the bound $\overline{\Vert \Delta(x) \Vert}$ decreases with decreasing $\overline{\Vert \epsilon(x) \Vert}$ and $\overline{\Vert \nabla \epsilon(x) \Vert}$.

\begin{thm}
Consider the extended system \eqref{extended_sys_disturbance}. Provided Assumption \ref{PE_condition} be satisfied. Under the adaptive critic updating law \eqref{critic law}-\eqref{Gamma law}, the augmented state $Z=\left[\zeta^\top,\tilde{W}^\top\right]^\top$ is stable in the sense of being uniformly ultimately bounded (UUB). 
\end{thm}
\begin{proof}
Consider the continuously differentiable Lyapunov function candidate
\begin{flalign}\label{lyapunov candidate}
& \mathcal{V}(Z,t) = V^*(\zeta) + V_c(\tilde{W},t), 
\end{flalign}
where $V^*(\zeta)$ is the optimal value function and $V_c(\tilde{W},t) = \frac{1}{2}\tilde{W}^\top \Gamma^{-1}(t)\tilde{W}$.

Taking the time derivative of $V^*(\zeta)$ along the system trajectories yields 
$\dot V^* = \nabla V^{*\top}(F_\zeta + G_\zeta\hat{\mu} + \Omega_\zeta \hat d)$.
Adding and subtracting the optimal policies $(\mu^*,d^*)$ and using the HJI equation gives
$\dot V^* = -\mathcal Q_\zeta(\zeta)-\mathcal R_\zeta(\mu^*)+\gamma^2\|d^*\|^2 
+ \nabla V^{*\top}G_\zeta(\hat\mu-\mu^*) 
+ \nabla V^{*\top}\Omega_\zeta(\hat d-d^*)$.
Using $\nabla V^{*\top}G_\zeta=-\mu^{*\top}R_\zeta$ and 
$\nabla V^{*\top}\Omega_\zeta=2\gamma^2 d^{*\top}$ together with Young's inequality yields
\begin{flalign}\label{vtm}
\dot{V}_{\zeta}\leq &-\mathcal{Q}_{\zeta}(\zeta)  +2\gamma^2 \| d^* \|^2 \notag\\
&+\frac{1}{2}\left\|  R_{\zeta}^{\frac{1}{2}}(\hat{\mu}-\mu^*) \right\|^2 + \gamma^2 \left\| \hat d - d^* \right\|^2. 
\end{flalign}
From \eqref{ideal_policy}--\eqref{ideal_disturbance} and 
\eqref{es_policy}--\eqref{es_disturbance}, the Nash equilibrium approximation errors satisfy
$\hat\mu-\mu^* = R_\zeta^{-1}G_\zeta^\top(-\nabla\sigma^\top\tilde W+\nabla\epsilon)$ and
$\hat d-d^* = \frac{1}{2\gamma^2}\Omega_\zeta^\top(\nabla\sigma^\top\tilde W-\nabla\epsilon)$.
Substituting these expressions into \eqref{vtm} and applying Young's inequality leads to
\begin{flalign}\label{dot_optimal_V_final}
\dot V^* \le &-\mathcal Q_\zeta(\zeta) + 2\gamma^2\|d^*\|^2
+ \varpi_1\|\tilde W\|^2 + \varpi_2\overline{\|\nabla\epsilon\|}^2, 
\end{flalign}
where $\varpi_1 = \frac{1}{2\gamma^2}\overline{\|\Omega_{\zeta}\|}^2\overline{\| \nabla\sigma \|}^2+\| R_{\zeta}^{-1} \| \overline{\| G_{\zeta} \|}^2\overline{\| \nabla\sigma \|}^2$ and $\varpi_2 = \frac{1}{2\gamma^2}\overline{\|\Omega_{\zeta}\|}^2 + \| R_{\zeta}^{-1} \| \overline{\| G_{\zeta} \|}^2$.

The time derivative of $V_c(\tilde{W},t)$ is given by $\dot{V}_c(\tilde{W},t) = \tilde{W}^\top \Gamma^{-1}\Big(W-\dot{\hat{W}}\Big) - \frac{1}{2}\tilde{W}^\top \Gamma^{-1}\dot{\Gamma}\Gamma^{-1} \tilde{W}$. Substituting the critic update law \eqref{critic law}-\eqref{Gamma law} into it yields $\dot{V}_c=-\frac{1}{2}\tilde{W}^\top \Big( \eta\Gamma^{-1}+\frac{k_1\rho\rho^\top}{\xi^2}+\frac{k_2}{N}\sum_{i=1}^N \frac{\rho_i\rho_i^\top}{\xi_i^2} \Big)\tilde{W}+\tilde{W}^\top\Gamma^{-1}W$. According to Assumption \ref{PE_condition}, we can further simplify it to

\begin{flalign}\label{result_Vc}
 \dot{V}_c(\tilde{W},t) &\leq -\frac{1}{2}\Big( \frac{\eta}{\overline{\Gamma}}+b \Big)\| \tilde{W} \|^2 + \frac{1}{\overline{\Gamma}}\overline{\| W \|}\| \tilde{W} \| + \Delta_{w}(x) \notag\\
&\leq -\frac{1}{2}\varpi_3 \| \tilde{W} \|^2 + \frac{\overline{\| W \|}^2}{2\overline{\Gamma}^2\varpi_3}+\Delta_w(x), 
\end{flalign}
where $\Delta_w(x):\mathbb{R}^n\to\mathbb{R}$ is uniformly bounded over $\mathcal{S}$ and consists of terms related to $\Delta(x)$ whose bound decreases with decreasing $\overline{\|\epsilon(x)\|}$ and $\overline{\|\nabla\epsilon(x)\|}$.

Recall the result for the time derivative of $V^*(\zeta)$. Along with the time derivative \eqref{result_Vc}, we have that
\begin{flalign}
\dot{\mathcal{V}}(Z)&\leq - \mathcal{Q}_{\zeta}(\zeta) - \Big(\frac{\varpi_3}{2}-\varpi_1\Big)\|\tilde{W}\|^2 + \varpi_4 \notag\\
&\leq - \lambda_{\min}(\Upsilon) \| Z\|^2 + \varpi_4, 
\end{flalign}
where $\varpi_4 = \varpi_2\overline{\| \nabla\epsilon\|}^2+\frac{\overline{\| W \|}^2}{2\overline{\Gamma}^2\varpi_3}+\Delta_w+2\gamma^2\overline{\| d^* \|}^2$ and $\Upsilon={\rm diag}(Q_{\zeta},\frac{\varpi_3}{2}-\varpi_1)$.
Hence, we have that
\begin{flalign}
& \dot{\mathcal{V}} \leq -\frac{\lambda_{\min}(\Upsilon)}{2}\Vert Z \Vert^2,\ \forall \Vert Z \Vert \ge \frac{2\varpi_4}{\lambda_{\min}(\Upsilon)}. 
\end{flalign}
The Lyapunov candidate $\mathcal{V}(Z,t)$ can be bounded as $\underline{v}(\Vert Z \Vert)\leq\mathcal{V}\leq\overline{v}(\Vert Z \Vert)$, where $\underline{v}(\cdot)$ and $\overline{v}(\cdot)$ are class $\mathcal{K}$ functions. According to \cite[Theorem 4.18]{khalil2002nonlinear}, we can conclude that the extended state $\zeta$ and the estimation error of critic weights $\tilde{W}$ are UUB. Specifically, all trajectories satisfy $\limsup_{t\to\infty}{\Vert Z \Vert \leq \overline{v}^{-1}\left( \underline{v}\left( \frac{2\varpi_4}{\lambda_{\min}(\Upsilon)} \right) \right)}$.
\end{proof}

\section{Simulation Results and Discussions}

This section presents three sets of simulations to validate the proposed framework. 
First, we illustrate the local trapping phenomenon caused by the interaction between safety constraints and optimal stabilization, and show that tangential excitation can drive the system away from such behavior. 
Second, for a nonlinear system with relative degree one, we demonstrate the tradeoff between safety enforcement and optimal stabilization performance. 
Third, for high-relative-degree systems under disturbances, we verify the effectiveness of the barrier-compensated high-order safe set and demonstrate how the adaptive $\lambda$ multiplier dynamically adjusts the safety weight to balance performance and safety.

\subsection{Tangential force as an $\mathcal{L}_2$ disturbance for escaping local trapping}

This subsection illustrates how an \emph{adaptive tangential force} can help the closed-loop system escape local trapping that may arise when the safeguarding mechanism dominates the control action near the boundary of a nonconvex safe set. The tangential force is treated as a constructed disturbance signal in the $L_2$ sense and is used to provide \emph{safe excitation} along BLF level sets, while preserving forward invariance of the safe set.

We consider the verification plant given by the two-dimensional integrator $\dot{x}=u$ with initial condition $x(0)=[-4,\,4]^\top$, simulation horizon $T=10\,\mathrm{s}$, and step size $\Delta t = 10^{-3}$. The quadratic cost matrices are $Q=\mathrm{diag}(2,2)$ and $R=\mathrm{diag}(1,1)$. Safety is defined as remaining outside the union of two disks that form a lens-shaped forbidden region. 
Let $x = [x_1,x_2]^\top$. The two disk centers are $c_{1,2} = c \pm d R(\theta)e_1$, where $c=[-2,\,2]^\top$, 
$d=1.2$, $\theta=45^\circ$, $e_1=[1,0]^\top$, and $R(\theta)$ is the planar rotation matrix; both disks share radius $r=1.0$. 
Define $h_i(x)=\|x-c_i\|^2-r^2$, $i=1,2$, so that $h_i(x)\ge0$ corresponds to being outside disk $i$. 
To obtain a smooth aggregate margin that approximates the non-smooth constraint 
$h(x)=\min\{h_1(x),h_2(x)\}$, we use the soft-min operator
\begin{flalign}\label{softmin}
& \operatorname{softmin}_\phi(a_1,\dots,a_m)
= -\frac{1}{\beta}\ln\!\left(\sum_{i=1}^m e^{-\beta a_i}\right), 
\end{flalign}
and define $h(x)=\operatorname{softmin}_\beta\big(h_1(x),h_2(x)\big)$ with $\phi=12$. 
The resulting safe set is $\mathcal{S}=\{x\mid h(x)\ge0\}$.

The BLF is constructed as
\begin{subequations}
\begin{flalign}
B(x) &= \frac{y(x)}{h(x)}, y(x) = 1 - \exp\{-a_y\|x\|^{p_y}\}, 
\label{BLF_y}
\end{flalign}    
\end{subequations}
where parameters $(a_y,p_y)=(0.1,0.5)$. The safeguarding controller is defined as $u_s(x)=-\lambda\nabla B(x)$. 
The critic is updated online using the Hamiltonian residual under the safe-exploration modification, where $(u_o+u_s)$ enters the dynamics term. The critic learning parameters are $(k_1,k_2,\eta,k_3)=(0.03,0.03,0.05,1.0)$ with $\Gamma(0)=I$, and the concurrent learning uses $N=6$ sampled trajectories.

We first demonstrate that, without tangential excitation, the closed-loop system may exhibit local trapping induced by the nonconvex safety geometry. 
With a small safeguarding gain $\lambda=0.2$, the trajectory becomes trapped near the lens region and converges to a non-origin equilibrium (see Fig.~\ref{fig:no_tan_multiIC}). 
In contrast, increasing the safeguarding gain to $\lambda=2.0$ enables the system to escape from the spurious equilibrium and converge toward the origin (see Fig.~\ref{fig:no_tan_compare}). Although the infinite-horizon formulation theoretically accounts for safety and should avoid spurious equilibria, the critic in this experiment is trained online. If the trajectory does not provide sufficient excitation around the safety boundary, the critic may fail to learn the embedded safety information in time. Consequently, the nominal optimal action and the safeguarding action may balance each other in a nonconvex region, yielding a locally consistent but globally suboptimal equilibrium. Increasing $\lambda$ strengthens the safety influence earlier and enriches the safety-relevant data observed by the critic, enabling the learned weights $W$ to encode safety information more effectively and escape the local trapping.

\begin{figure}[t]
  \centering
  \subfloat[$\lambda=0.2$, multiple initial conditions.]{
      \includegraphics[width=0.46\linewidth]{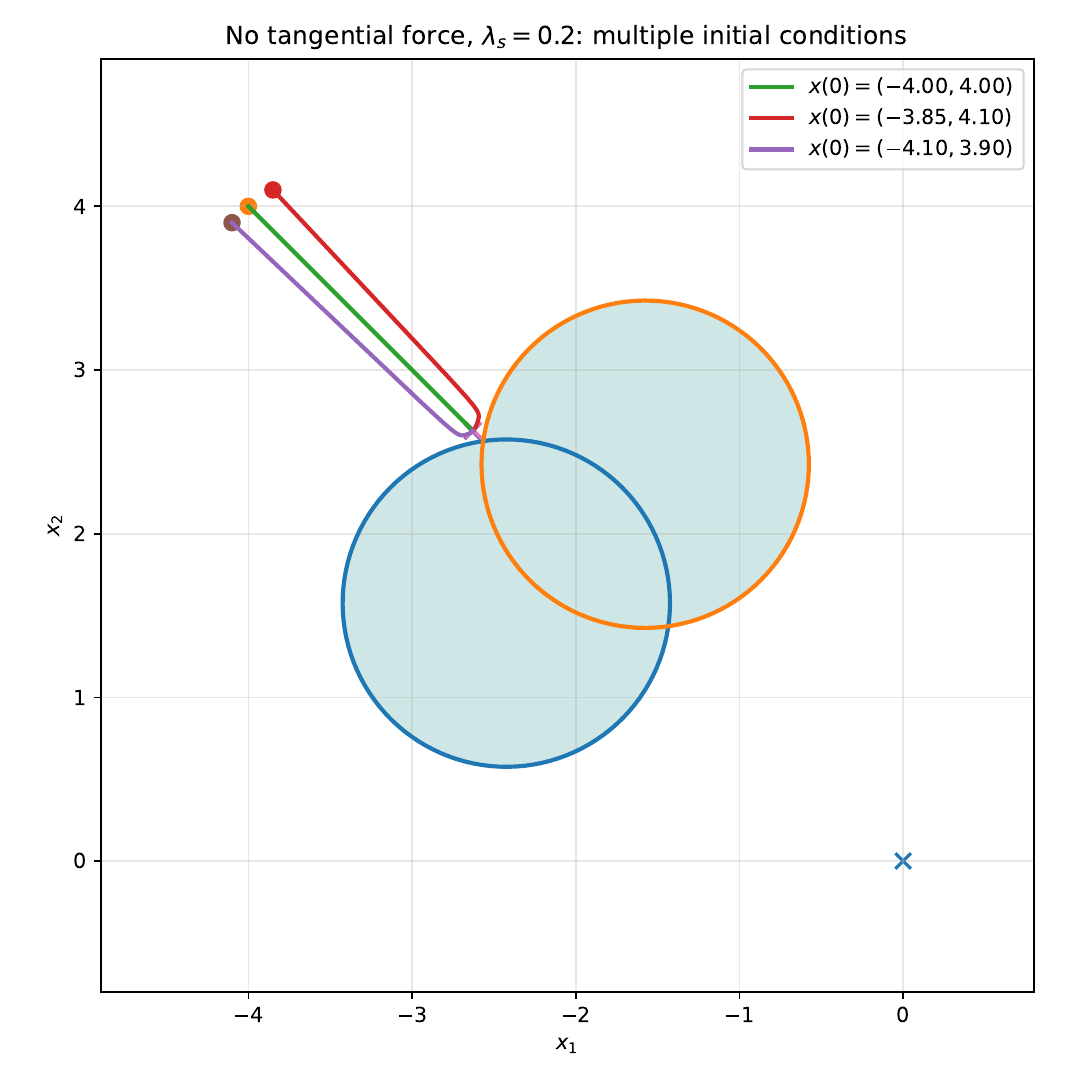}
      \label{fig:no_tan_multiIC}
  }
  \hfill
  \subfloat[$\lambda=0.2$ vs.\ $2.0$.]{
      \includegraphics[width=0.46\linewidth]{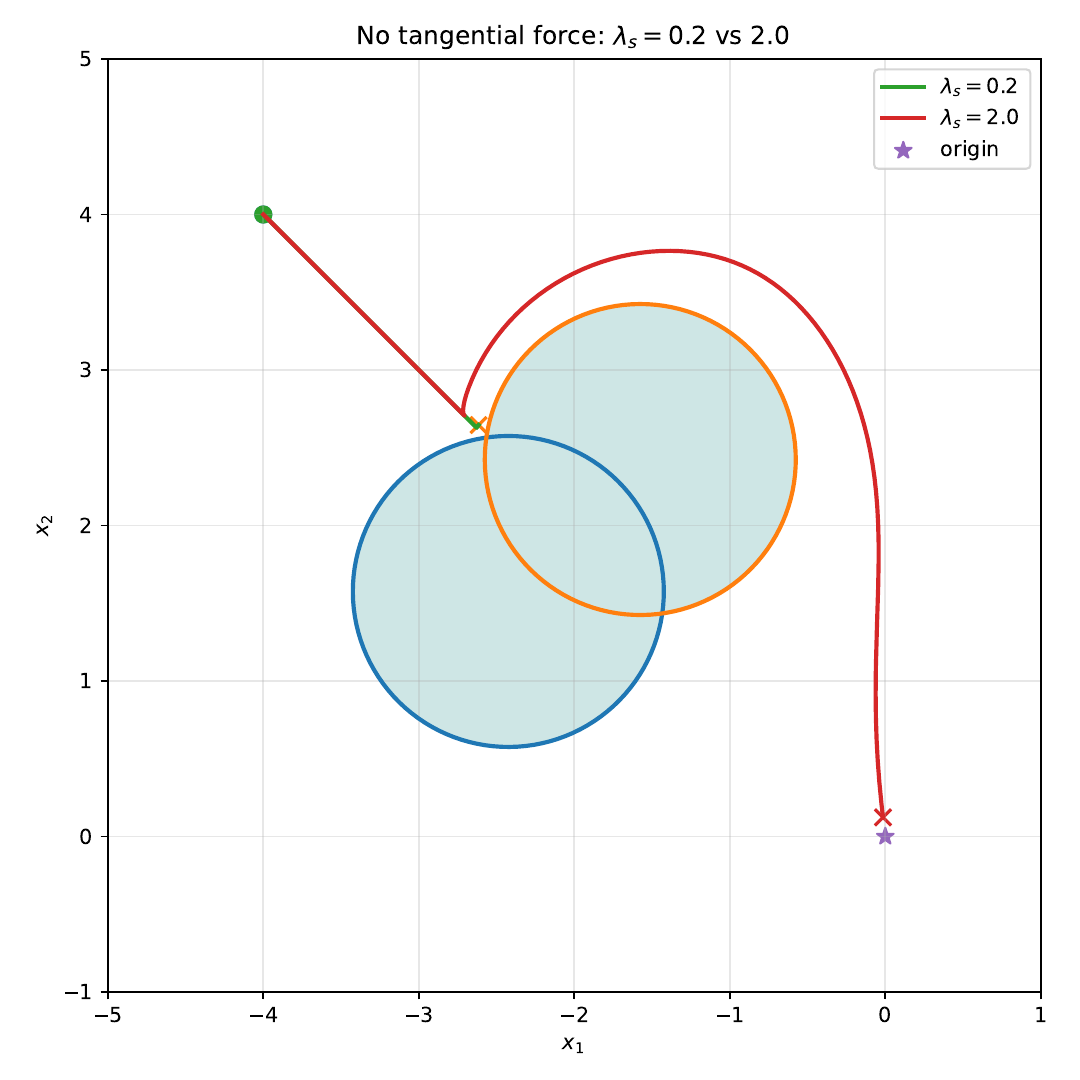}
      \label{fig:no_tan_compare}
  }
  \caption{Effect of safeguarding gain without tangential excitation. 
  Left: trajectories exhibit local trapping under small $\lambda$. 
  Right: increasing $\lambda$ enables convergence to the origin.}
  \label{fig:no_tan_combined}
\end{figure}

To eliminate the local trapping behavior observed in the previous counterexample, we activate the adaptive $\mathcal{L}_2$ tangential excitation introduced in the theoretical section. 
Recall that the tangential term takes the form $u_t(x)=p(x)t(x)$, where $t(x)$ denotes a direction orthogonal to the barrier gradient $\nabla B(x)$. 
For the two-dimensional case, the orthogonal direction can be obtained using the rotation matrix $t(x)
=
\frac{R_t\,\nabla B(x)}
{\|\nabla B(x)\|},
R_t=
\begin{bmatrix}
0 & -1 \\
1 & 0
\end{bmatrix}.$ By construction, $t(x)^\top\nabla B(x)=0$, so the tangential input does not directly affect the barrier decrease condition and therefore preserves the safety geometry. 
Instead, it induces motion along the level sets of the BLF, providing exploration directions that are not aligned with the barrier gradient. The excitation gain $p(x)$ is generated by the first-order filter described as \eqref{excitation_gain_p}. The parameter values used in the simulation are chosen as
$\alpha_t=2$, $k_t=0.5$, $\varepsilon=10^{-4}$, and the initial condition of the filter is $p(0)=0$. Since the tangential input is treated as a constructed disturbance channel, the critic update adopts the $\mathcal{H}_\infty$ formulation. The infinite-horizon objective is given by $V(x_0)=\int_0^\infty \big(x^\top Q x + \frac{1}{2}u^\top R u - \gamma^2 d^\top d\big)\,dt$, with the exponential decay factor $\gamma=2$. 

Figure~\ref{fig:tan_full} shows that, although the safeguarding gain remains $\lambda=0.2$, the introduction of the adaptive $\mathcal{L}_2$ tangential excitation enables the trajectory to leave the trap region and converge to the origin while preserving $h(x(t))\ge0$. 
Compared with the disturbance-free case, the system no longer settles near the boundary-induced equilibrium. 
The quiver arrows indicate that the tangential action becomes significant during the boundary-interaction phase and gradually diminishes afterward. The time history of $p(x)$ confirms that the excitation decays toward zero once the trajectory moves away from the boundary-dominated regime, consistent with the $\mathcal{L}_2$ design requirement. 
Meanwhile, the critic weights $W(t)$ exhibit pronounced variation during the interaction phase and subsequently stabilize, indicating that the tangential motion enriches the online data and allows the embedded safety information to be learned effectively. 

\begin{figure}[t]
  \centering
  \begin{minipage}[t]{0.48\linewidth}
    \vspace{0pt}
    \centering
    \includegraphics[width=\linewidth]{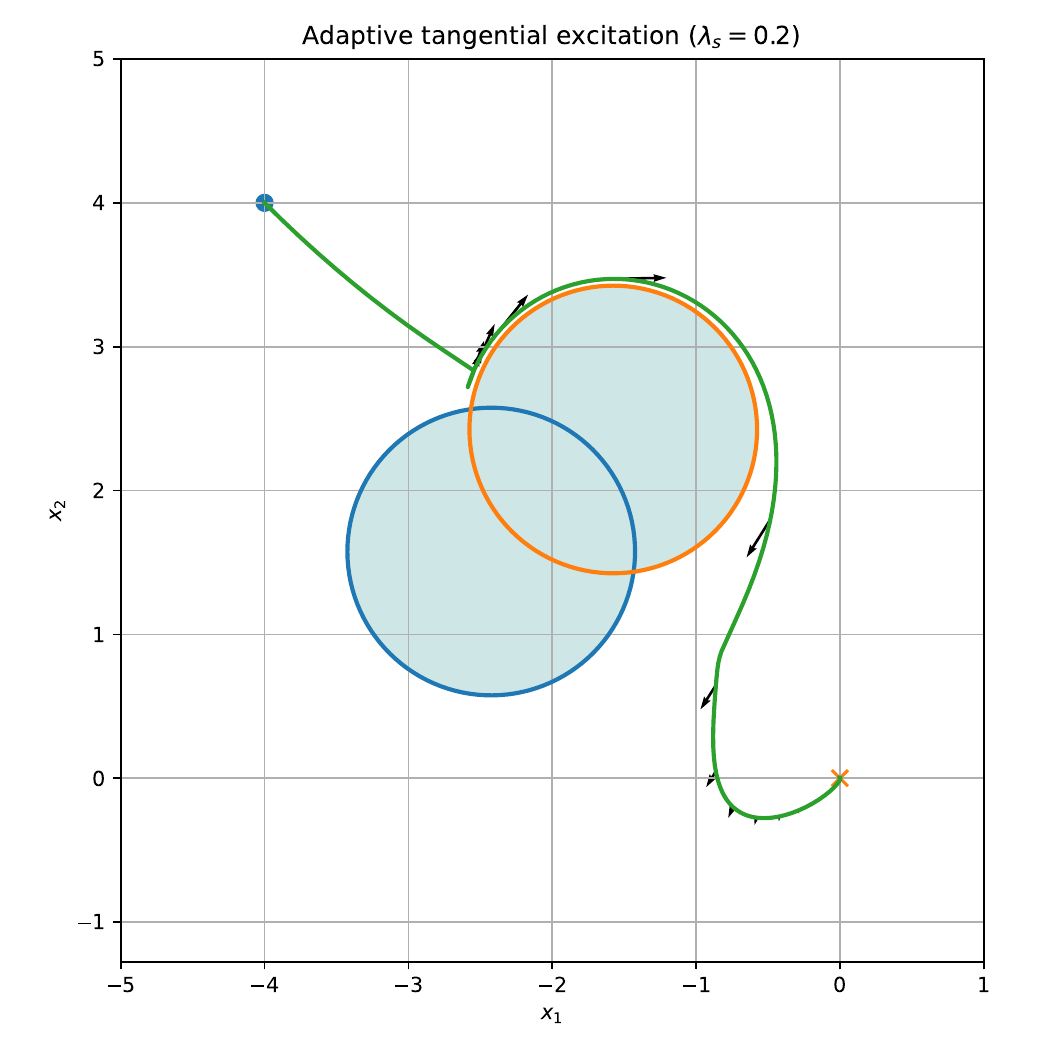}
  \end{minipage}
  \hfill
  \begin{minipage}[t]{0.48\linewidth}
    \vspace{0pt}
    \centering
    \includegraphics[width=\linewidth]{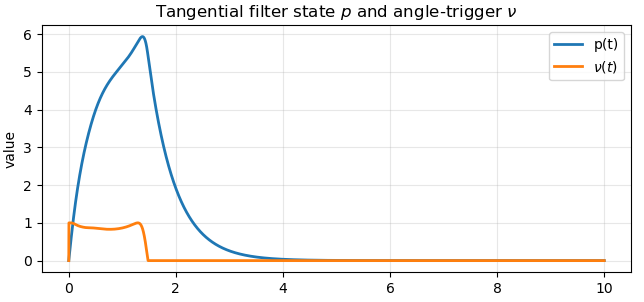}
    
    \vspace{2mm}
    
    \includegraphics[width=\linewidth]{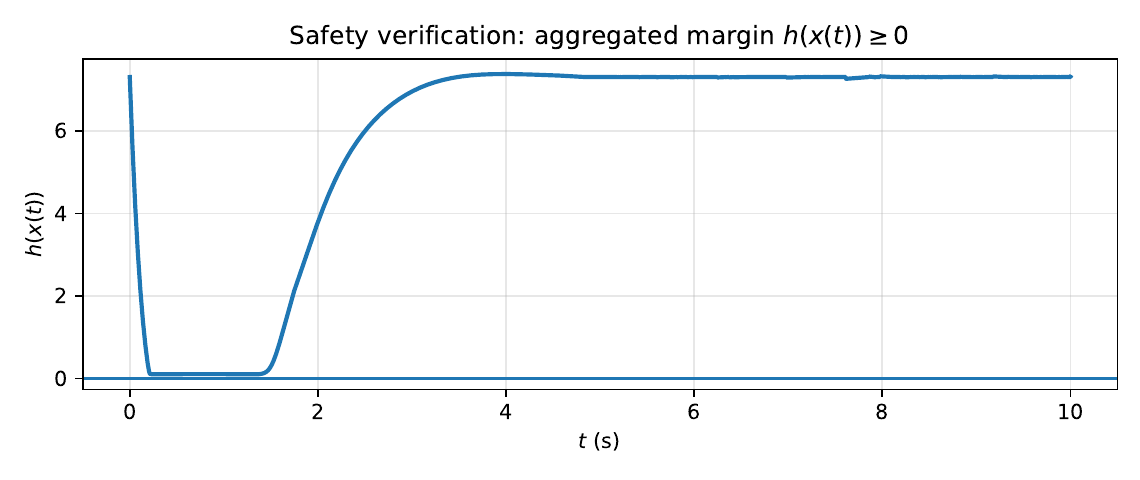}
  \end{minipage}
  
  \vspace{3mm}
  
  \includegraphics[width=0.7\linewidth]{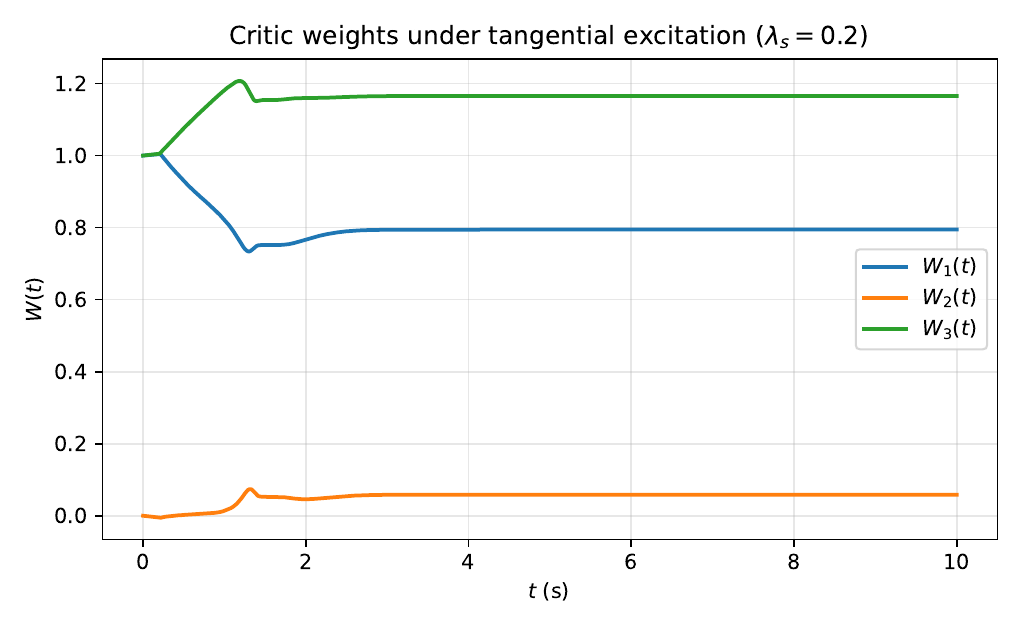}

  \caption{Adaptive tangential excitation under a weak safeguarding gain ($\lambda_s=0.2$). 
  Top-left: state trajectory with tangential disturbance arrows $u_t$. 
  Top-right: excitation dynamics ($z(t)$ and trigger $\psi(t)$) and the aggregated safety margin $h(x(t))$. 
  Bottom: evolution of critic weights $W(t)$.}
  \label{fig:tan_full}
\end{figure}

\subsection{Nonlinear systems with relative-degree-one}

In this subsection, we consider a nonlinear system with relative degree one under bounded disturbance given by 
$\dot{x}=f(x)+u-d$, where $x=[x_1,x_2]^\top\in\mathbb{R}^2$ and 
$f(x)=\begin{bmatrix}\tanh(x_1+x_2)x_2\\-\tanh(x_1+x_2)x_1\end{bmatrix}$. The disturbance channel is treated within the $H_\infty$ framework, and the adversarial input is generated by the Isaacs worst-case policy $d^\star=-\frac{1}{2\gamma^2}\nabla V^*(\zeta)$, where $\gamma=2$ specifies the disturbance attenuation level. The resulting disturbance signal is shown in Fig.~\ref{fig:rel1_disturbance}, where it remains bounded throughout the simulation horizon.

The safety constraint is defined by a circular obstacle, and the safe set is characterized by $h(x)=(x_1+3.5)^2+(x_2+0.5)^2-1 \ge 0$. The BLF is constructed as \eqref{BLF_y}. The safeguarding controller is chosen as $u_s=-\lambda\nabla B(x)$, and the applied control is decomposed as $u=u_o+u_s$, where $u_o$ is produced by the learned $H_\infty$ policy (see Fig.~\ref{fig:rel1_control_decomp}).

For this first-order system, $g(x)=I$ holds, which implies $\mathcal{L}_gB(x)=\nabla B(x)$. Consequently, the safeguarding action is directly aligned with the gradient direction of the barrier function. This structural property yields an inherent robustness feature: even in the presence of adversarial disturbance, the safeguarding term counteracts the tendency of the state to move toward the unsafe set. As observed in Fig.~\ref{fig:rel1_traj_compare}, no safety violation occurs and the trajectory maintains a positive geometric margin from the obstacle. This confirms that, in relative-degree-one settings, the BLF-based safeguarding controller exhibits natural disturbance rejection capability.

The simulation horizon is $T=25\,\mathrm{s}$ with time step $\Delta t=10^{-3}$. The initial condition is $x(0)=[-5,4]^\top$ and $\lambda(0)=5$. The quadratic cost matrices are chosen as $Q=\mathrm{diag}(2,2)$ and $R=\mathrm{diag}(1,1)$. The critic is initialized with $W(0)=[0.5,1.0,0.2,0,0,0.2]^\top$ and $\Gamma(0)=I$, and the update gains are selected as $k_1=0.05$, $k_2=0.02$, $\eta=0.01$, and $k_3=1$. Ten safe exploration-based sampled trajectories are drawn at each step for concurrent learning.

To improve smoothness and reduce conservatism, the Lagrange multiplier is treated as a dynamic state governed by $\dot{\lambda}=v$. A state-dependent lower reference is introduced,
\begin{flalign}\label{lambda_ref}
\lambda_{\mathrm{ref}}(x)=\lambda_{\min}+\Delta\lambda\,
\cdot\mathrm{sigmoid}\!\left(\frac{h_0-h(x)}{\Delta h}\right), 
\end{flalign}
where $\mathrm{sigmoid}(z):=(1+\exp(-z))^{-1}$ is the logistic sigmoid function, $\lambda_{\min}=0.2$, $\Delta\lambda=4.0$, $h_0=1.0$, and $\Delta h=3.0$. This reference acts as an early-warning threshold. The CBF constraint enforces $\lambda(t)\ge\lambda_{\mathrm{ref}}(x(t))$ and $\lambda(t)>0$ by restricting the virtual input $v$. Meanwhile, the performance objective incorporates the terms $q_\lambda(\lambda-\lambda^\star)^2+r_v v^2$ in the running cost, where $\lambda^\star=1$, $q_\lambda=2$, and $r_v=0.5$.

In the critic approximation, the augmented state is $\zeta=[x^\top,\tilde{\lambda}]^\top$, and the basis is constructed using $\tilde\lambda=\lambda-\lambda^\star$ to avoid bias from the steady-state value. Quadratic features of $(x,\tilde\lambda)$ are adopted to approximate the value function.

The learning process is illustrated in Fig.~\ref{fig:rel1_W_delta_strip}. The critic weights remain bounded and converge smoothly, while the Isaacs residual magnitude $\log|\delta(t)|$ decays rapidly, indicating that the Hamiltonian error diminishes as learning progresses. This confirms stable adaptation on the extended system.

The trajectory comparison in Fig.~\ref{fig:rel1_traj_compare} highlights the benefit of adaptive $\lambda$. In the constant-multiplier case, selecting $\lambda$ involves a trade-off between robustness and smoothness. In contrast, the adaptive scheme increases $\lambda$ only when approaching the constraint boundary, as reflected by the color-coded multiplier values, and relaxes it afterward. The detailed evolution of $\lambda(t)$ and its lower reference $\lambda_{\mathrm{ref}}(x(t))$ is shown in Fig.~\ref{fig:rel1_lambda_track}, where $\lambda(t)$ rises in advance of the boundary and subsequently returns toward $\lambda^\star$ once sufficient safety margin is restored.


\begin{figure}[t]
  \centering
  \includegraphics[width=\linewidth]{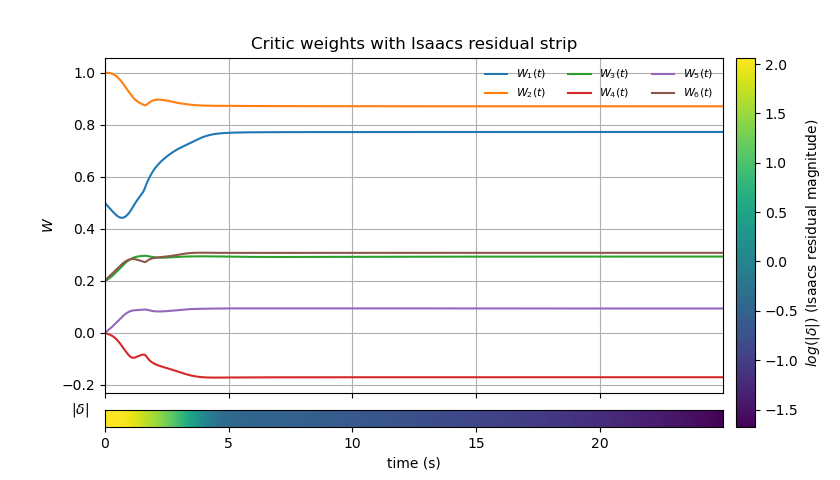}
  \caption{Critic weight evolution with Isaacs residual strip for the first-order system.
  The upper panel shows the evolution of critic weights $W_i(t)$.
  The color strip beneath encodes the magnitude of the Hamiltonian/Isaacs residual,
  represented as $\log |\delta(t)|$.
  The rapid fading of the color intensity indicates the decay of the residual,
  consistent with critic convergence.}
  \label{fig:rel1_W_delta_strip}
\end{figure}

\begin{figure}[t]
  \centering
  \includegraphics[width=0.8\linewidth]{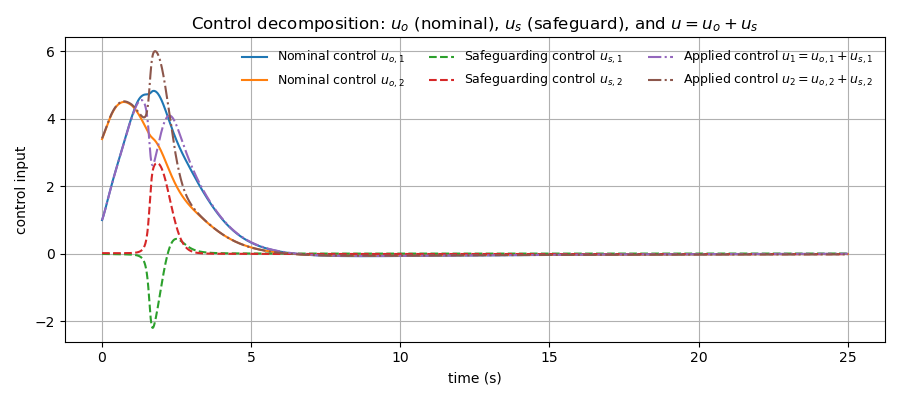}
  \caption{Control decomposition. The nominal control $u_o$ is produced by the learned $H_\infty$ policy, while the safeguarding control $u_s=-\lambda\nabla B(x)$ enforces safety. The applied input is $u=u_o+u_s$.}
  \label{fig:rel1_control_decomp}
\end{figure}

\begin{figure}[t]
  \centering
  \includegraphics[width=0.8\linewidth]{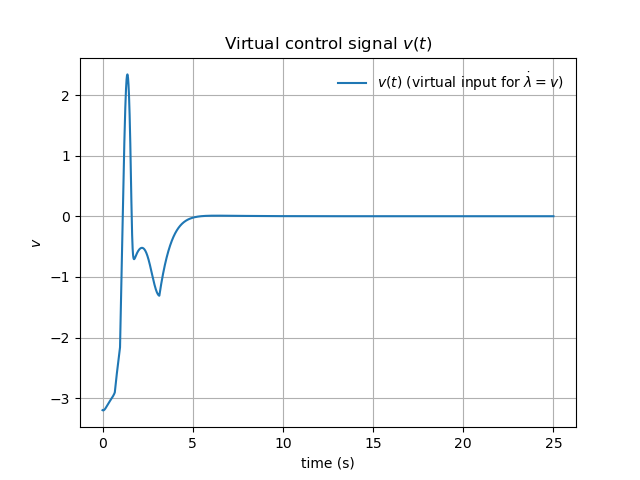}
  \caption{Virtual control input $v(t)$ for the multiplier dynamics. The auxiliary state $\lambda$ evolves according to $\dot{\lambda}=v$.}
  \label{fig:virtual_control_input}
\end{figure}

\begin{figure}[t]
  \centering
  \includegraphics[width=0.8\linewidth]{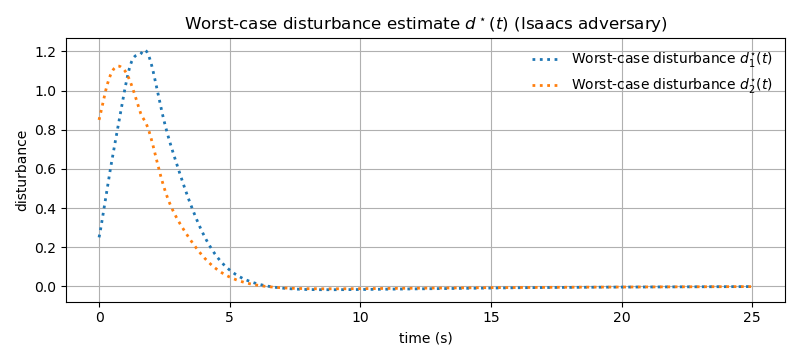}
  \caption{Worst-case disturbance estimate $d^\star(t)$ given by the Isaacs adversary in the $H_\infty$ formulation.}
  \label{fig:rel1_disturbance}
\end{figure}

\begin{figure}[t]
  \centering
  \includegraphics[width=0.8\linewidth]{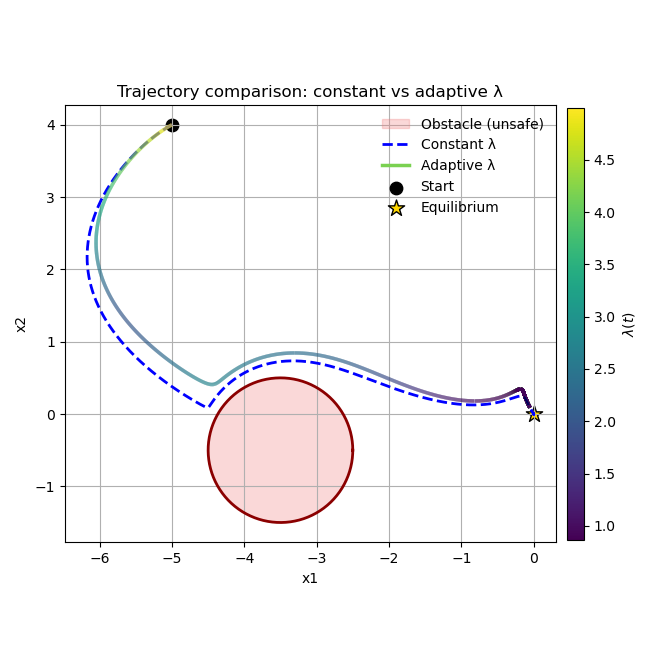}
  \caption{Trajectory comparison between constant and adaptive $\lambda$. The unsafe region (obstacle interior) is shaded. The adaptive-$\lambda$ trajectory is color-coded by $\lambda(t)$, indicating stronger safeguarding near the constraint.}
  \label{fig:rel1_traj_compare}
\end{figure}

\begin{figure}[t]
  \centering
  \includegraphics[width=0.8\linewidth]{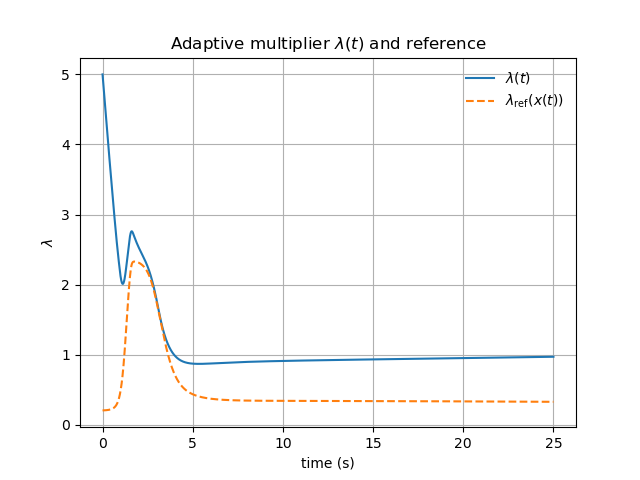}
  \caption{Adaptive multiplier $\lambda(t)$ and its state-dependent lower reference $\lambda_{\mathrm{ref}}(x(t))$ (implemented via a CBF lower-bound). The multiplier increases when approaching the constraint and relaxes afterward.}
  \label{fig:rel1_lambda_track}
\end{figure}

\subsection{High-relative-degree mobile robots in minefields}

We next evaluate the proposed safe exploration enhanced learning framework on a high-relative-degree plant in a minefield environment. The verification model is a planar double integrator with an additive flow-field disturbance,
$\dot p = v + d$ and $\dot v = u$,
where $p\in\mathbb{R}^2$ and $v\in\mathbb{R}^2$ denote position and velocity, respectively. The disturbance $d$ enters the position channel and represents a flow-field type perturbation. 
This disturbance is not explicitly canceled by the controller and is introduced to evaluate robustness against modeling mismatch.

The controller is implemented with an $\mathcal{H}_\infty$ critic-learning objective and an adaptive safeguarding multiplier. The running cost uses $Q=\mathrm{diag}(2.0,2.0,0.4,0.4)$ and $R=\mathrm{diag}(1.0,1.0)$, while the $\mathcal{H}_\infty$ attenuation level is set to $\gamma=4.0$. The critic weights are initialized from an LQR fit at the origin, and the continuous-time critic update laws are discretized by Euler integration. The key learning parameters are $k_1=0.1$, $k_2=0.1$, $\eta_g=0.01$, and $k_3=0.1$, with online sampling size $N=10$ and sampling standard deviation $0.05$. 

The safe set is defined by an outer boundary and multiple circular obstacles randomly generated inside the boundary. Since the safety constraints depend only on position, the system has relative degree two with respect to the safety output. Accordingly, the high-order safe set is constructed through the second-order constraint $\psi_{i,2} = \dot{h}_i + k_1 h_i - \frac{\phi}{h_i}$, where $k_1=2$ and the additional term $-\phi/h_i$ with $\phi=0.2$ provides robustness against modeling errors and disturbances. To obtain a smooth aggregate margin for multiple constraints, the high-order terms are combined using the soft-min operator $\psi(x)
= -\frac{1}{\beta}
\ln\!\left(
\sum_{i=1}^{m}
e^{-\beta \psi_{i,2}(x)}
\right),$ where $\beta=12$. The resulting robustness-enhanced high-order safe set is defined as $\mathcal{S}=\{x\mid \psi(x)\ge0\}$. Following the same construction as in the first-order case, the BLF is defined as $B(x) = \frac{y(x)}{\psi(x)}$, with $y(x) = 1 - \exp\{-0.1\|x\|^{0.5}\}$. The safeguarding controller is defined as $u_s(x)=-\lambda\nabla B(x)$.

The trajectory comparison under the same disturbance realization is shown in Fig.~\ref{fig:minefield_robust_compare}. 
When the robustness term $\phi/h$ is not included, the trajectory (red solid curve) approaches the active constraint and exhibits a temporary safety violation. 
This phenomenon indicates that the nominal high-order condition $\dot{h}_i + k_1 h_i \ge 0$ is insufficient to compensate for disturbance-induced drift in the position dynamics. 
Consequently, forward invariance of the safe set cannot be guaranteed under perturbations.
In contrast, when the robustness enhancement term is incorporated, the trajectory (blue dashed curve) remains strictly within the admissible region. 
The inset provides a zoomed-in view near the active constraint, clearly illustrating the separation between the trajectory and the boundary. 
The term $-\phi/h(x)$ introduces a barrier-type amplification as $h(x)\to 0^+$, which counteracts the disturbance effect and restores robust forward invariance of the high-order safe set.


\begin{figure}[t]
\centering
\includegraphics[width=0.8\linewidth]{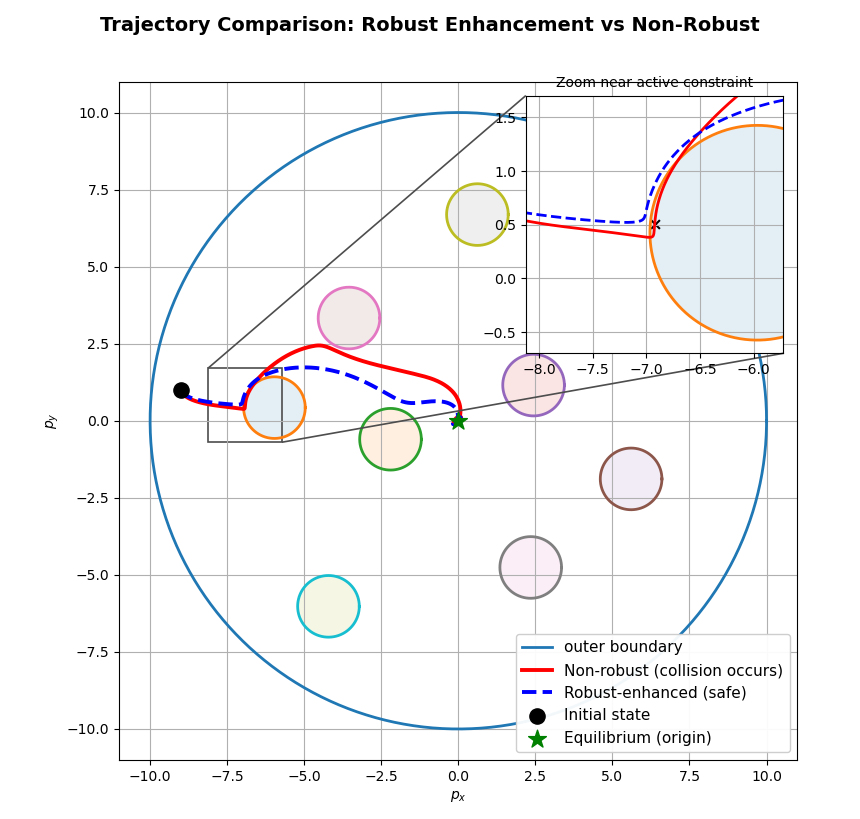}
\caption{
Trajectory comparison under the same disturbance realization.
The red solid curve represents the high-order safe set without robustness enhancement, 
which leads to constraint violation. 
The blue dashed curve represents the high-order safe set with the robustness enhancement term $-\phi/h(x)$, 
which maintains safety.
The inset shows a zoomed-in view near the active constraint.
}
\label{fig:minefield_robust_compare}
\end{figure}

\begin{figure}[t]
\centering
\includegraphics[width=0.8\linewidth]{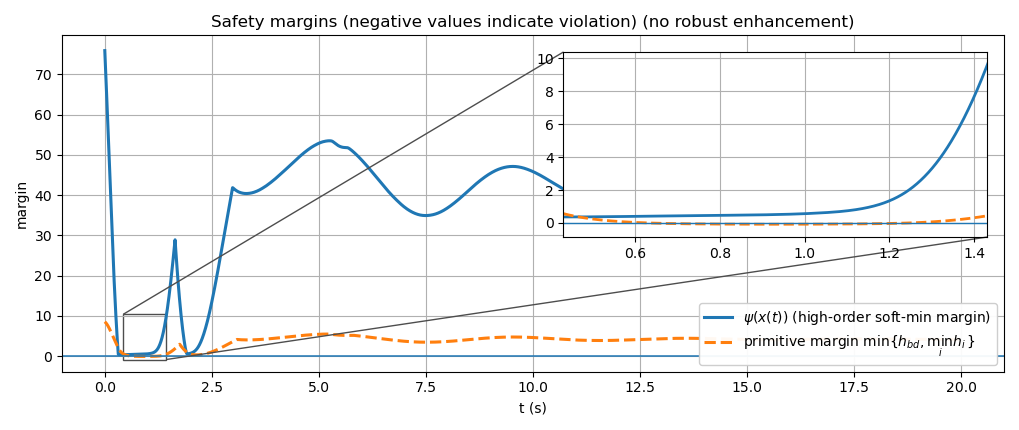}
\caption{
Safety margin evolution without robustness enhancement.
The solid blue curve represents the high-order soft-min safety margin $\psi(x(t))$,
while the dashed orange curve denotes the primitive margin 
$\min\{h_{\mathrm{bd}}, \min_i h_i\}$.
Negative values indicate constraint violation.
}
\label{fig:safety_margin_no_robust}
\end{figure}

To improve safety responsiveness while avoiding excessive conservatism, 
the safeguarding multiplier $\lambda$ is no longer selected as a constant 
but generated by an auxiliary dynamical system. Specifically, we introduce 
the first-order multiplier dynamics $\dot{\lambda}=v$, where $v$ is treated as a virtual control input. The state-dependent lower reference $\lambda_{\mathrm{ref}}(x)$ is constructed 
in the same manner as in the relative-degree-one case (see \eqref{lambda_ref}), 
which serves as an early-warning threshold that increases the safeguarding 
authority when the state approaches safety-critical regions. In the present high-order simulation, the parameters are selected as 
$\lambda_{\min}=0.5$, $\Delta\lambda=6.0$, $h_0=1.0$, and $\Delta h=3.0$.

To guarantee admissible multiplier values, CBFs are 
constructed for the auxiliary system to enforce both the lower constraint 
$\lambda \ge \lambda_{\mathrm{ref}}(x)$ and the upper constraint 
$\lambda \le \lambda_{\max}$ with $\lambda_{\max}=20$. The desired nominal multiplier is chosen as $\lambda^\star=2$, which 
represents a balanced safeguarding level observed empirically. To stabilize the auxiliary dynamics around this value, the HJI objective incorporates the quadratic penalty $q_{\lambda}(\lambda-\lambda^\star)^2+r_vv^2$, where $q_{\lambda}=1$ and $r_v=0.1$. This provides a trade-off between tracking the nominal multiplier and limiting rapid variations of $\lambda$.

Under the adaptive multiplier scheme, the state is augmented as 
$\zeta=[x^\top,\tilde{\lambda}]^\top$. The critic approximation is constructed on the extended state using a second-order polynomial basis.

\begin{figure}[t]
\centering
\includegraphics[width=0.78\linewidth]{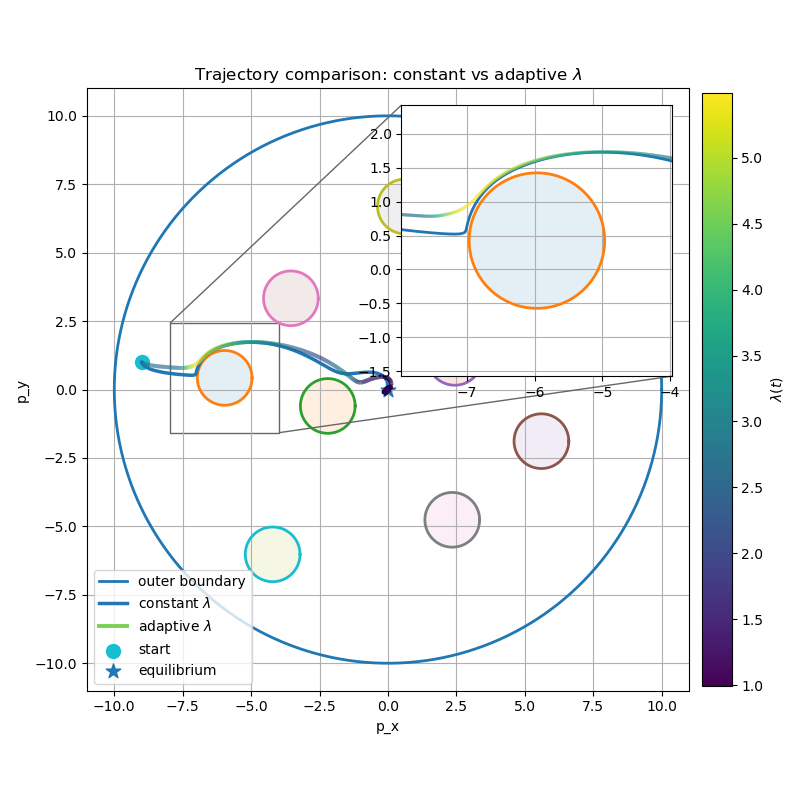}
\caption{
Trajectory comparison under constant and adaptive $\lambda$.
The solid curve corresponds to the constant $\lambda$ safeguarding strategy,
while the dashed curve represents the adaptive $\lambda$ strategy.
The inset provides a magnified view around the obstacle centered at 
$(-5.9660,\;0.4260)$, where the adaptive mechanism yields earlier deviation.
}
\label{fig:traj_constant_vs_adaptive_lambda}
\end{figure}

\begin{figure}[t]
\centering
\begin{minipage}[t]{0.48\linewidth}
    \centering
    \includegraphics[width=\linewidth,height=2.7cm]{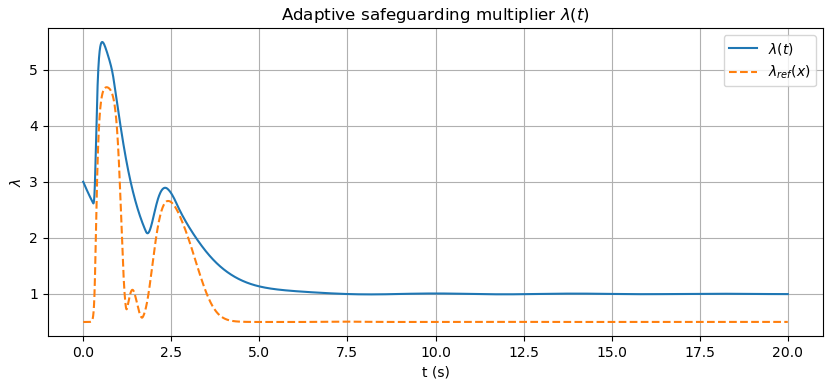}
\end{minipage}
\hfill
\begin{minipage}[t]{0.48\linewidth}
    \centering
    \includegraphics[width=\linewidth,height=2.7cm]{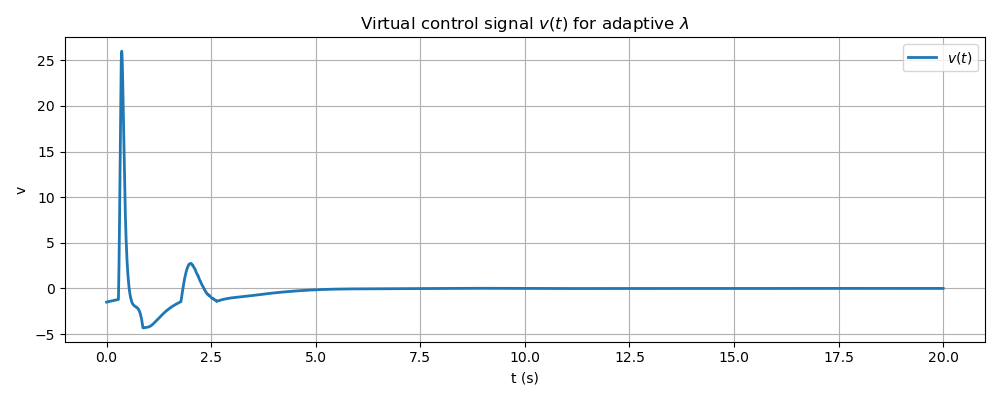}
\end{minipage}

\caption{
(a) Evolution of the adaptive safeguarding multiplier $\lambda(t)$
and its lower bound reference $\lambda_{\mathrm{ref}}(x)$.
(b) Virtual control input $v(t)$ governing the multiplier dynamics
$\dot{\lambda}=v$.
}
\label{fig:lambda_v}
\end{figure}

\begin{figure}[t]
\centering
\begin{minipage}[t]{0.48\linewidth}
    \centering
    \includegraphics[width=\linewidth,height=3.5cm]{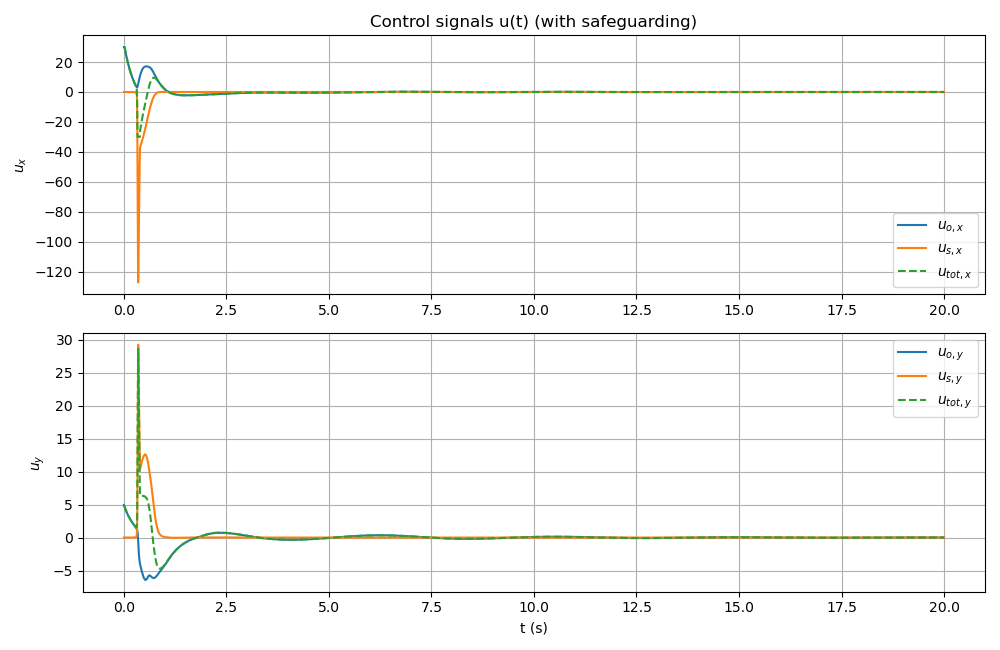}
\end{minipage}
\hfill
\begin{minipage}[t]{0.48\linewidth}
    \centering
    \includegraphics[width=\linewidth,height=3.5cm]{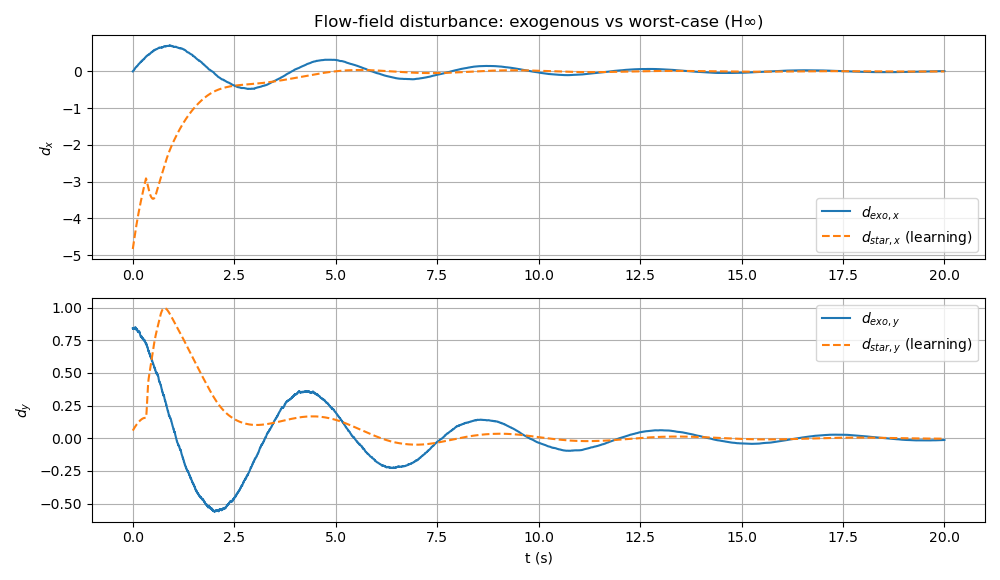}
\end{minipage}

\caption{
(a) Control signals under safeguarding. The total input is decomposed as
$u(t)=u_o(t)+u_s(t)$.
(b) Flow-field disturbance: exogenous signal $d_{\mathrm{exo}}(t)$ and
worst-case estimate $d^\star(t)$ from the $\mathcal{H}_\infty$ formulation.
}
\label{fig:control_disturbance}
\end{figure}

\begin{figure}[t]
\centering
\includegraphics[width=\linewidth]{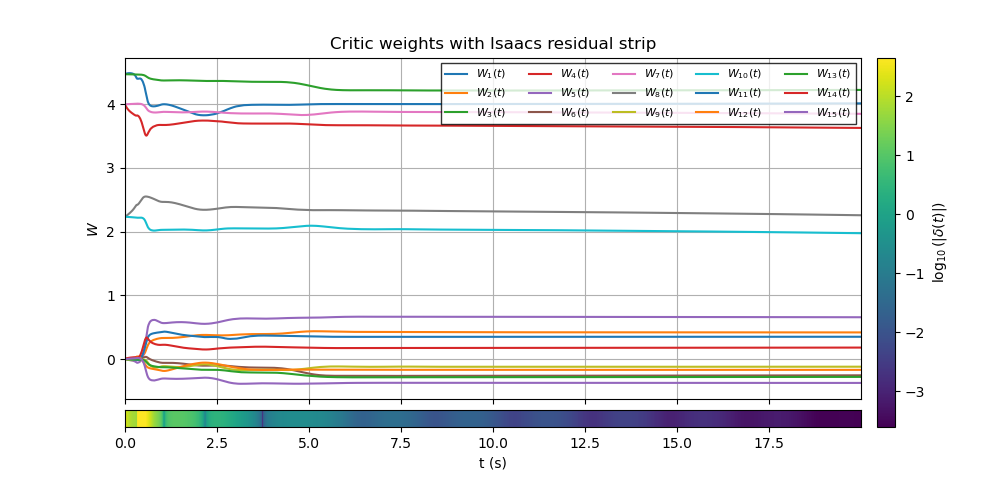}
\caption{
Evolution of critic weights $W(t)$ under safe exploration.
}
\label{fig:critic_weights}
\end{figure}

Fig.~\ref{fig:lambda_v}(a) shows the evolution of the adaptive multiplier.
When the system approaches the obstacle, 
$\lambda(t)$ rapidly increases beyond its nominal value, 
indicating early activation of the safeguarding mechanism.
As the state leaves the safety-critical region, 
$\lambda(t)$ smoothly decreases and converges to its steady value, 
demonstrating adaptive relaxation. Fig.~\ref{fig:lambda_v}(b) depicts the virtual control signal $v(t)$.
A positive peak appears when the system approaches the obstacle, 
driving $\lambda$ upward to strengthen the safeguarding effect.
Subsequently, the virtual control regulates the multiplier within the interval 
defined by the lower bound $\lambda_{\mathrm{ref}}(x)$ and the nominal value $\lambda^*$,
leading $\lambda(t)$ to gradually converge to $\lambda^*$. The resulting control signals and disturbance comparison are shown in Fig.~\ref{fig:control_disturbance}.
Compared with the constant-$\lambda$ case, the adaptive strategy modifies the safeguarding component only when necessary, thereby reducing unnecessary control amplification away from the constraint boundary. The disturbance comparison further illustrates the exogenous disturbance signal $d_{\mathrm{exo}}$ and the worst-case disturbance $d^\star$ arising from the $\mathcal{H}_\infty$ formulation. Fig.~\ref{fig:critic_weights} shows that the critic weights converge under the safe exploration framework. Meanwhile, the Hamiltonian error gradually decreases and approaches zero, indicating that the learning process successfully approximates the optimal solution. Compared with the constant multiplier strategy, the adaptive mechanism yields earlier deviation from the obstacle
while avoiding persistent over-conservatism in steady state.

\begin{figure}[t]
  \centering
  \includegraphics[width=0.8\linewidth]{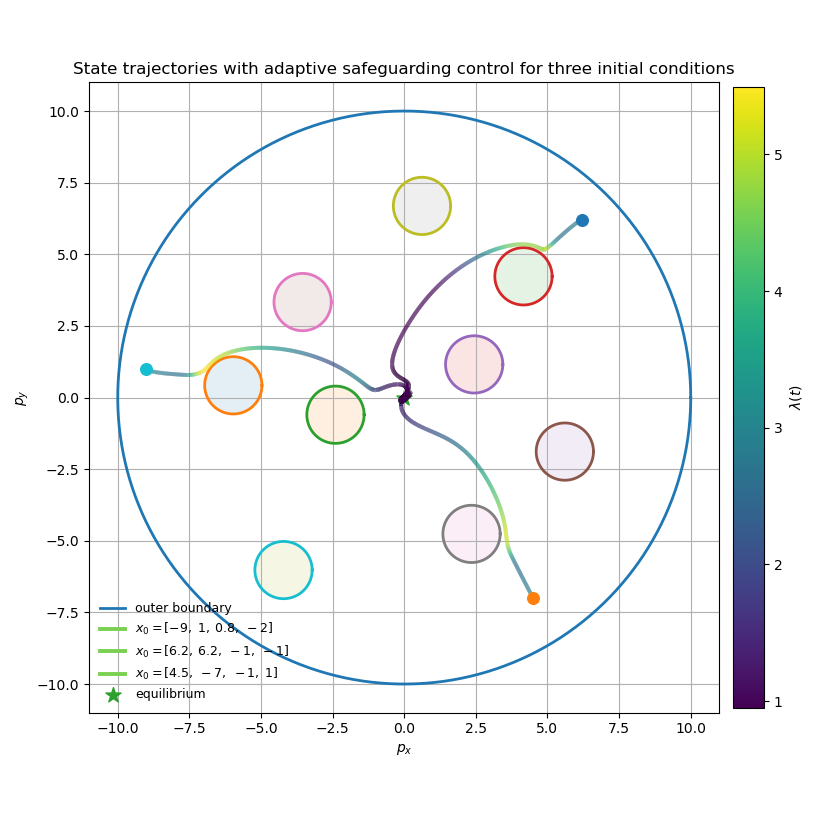}
  \caption{Minefield navigation for the high-relative-degree plant under flow-field disturbance. 
  The three trajectories start from $(-9.0,\,1.0,\,0.8,\,-2.0)$, $(6.2,\,6.2,\,-1.0,\,-1.0)$, and $(4.5,\,-7.0,\,-1.0,\,1.0)$, respectively.}
  \label{fig:case2_minefield_3starts}
\end{figure}

Fig.~\ref{fig:case2_minefield_3starts} shows the closed-loop trajectories from three distinct initial states,
$(-9.0,1.0,0.8,-2.0)$, $(6.2,6.2,-1.0,-1.0)$, and $(4.5,-7.0,-1.0,1.0)$.
In all cases, the system successfully navigates through the minefield and moves toward the origin without entering obstacles or violating the outer-boundary constraint. These results indicate that the proposed $\mathcal{H}_\infty$ critic learning, together with the safeguarding mechanism, can handle both the high-relative-degree dynamics and the presence of exogenous flow-field disturbances in a robust manner.

\section{Conclusion}
This paper developed a safety-aware infinite-horizon optimal control framework for nonlinear systems by embedding the BLF-based safeguarding action into the system dynamics and introducing a barrier-regulating auxiliary variable to adaptively adjust the safeguarding authority. This construction transforms the original safety-constrained problem into an unconstrained optimal control problem on an extended state space. To mitigate local trapping, an adaptive alignment-conditioned tangential excitation was introduced orthogonally to the safety direction and incorporated as an admissible $\mathcal{L}_2$ disturbance in the $H_\infty$ formulation without directly affecting the barrier decrease condition. For disturbed high-relative-degree systems, the recursive high-order safe-set construction was further augmented with barrier compensation terms, leading to a corresponding high-order BLF and a zero-sum differential game formulation. A safe-exploration-enhanced online critic learning scheme was employed for approximate solution of the resulting optimality conditions. Simulation results verified that the proposed method can mitigate local trapping, improve the safety--performance trade-off, and maintain safe operation under disturbances. Future work will address state-uncertainty-induced errors in BLF-based safeguarding control and investigate offline training schemes that reduce reliance on tangential excitation.

\bibliographystyle{IEEEtran}
\bibliography{reference}

\vfill

\end{document}